\documentclass[12pt]{article}

\usepackage{amsmath,bbm,amssymb}
\usepackage{amsthm}
\usepackage{comment}
\usepackage{graphicx, graphics}
\usepackage{url}
\usepackage{psfrag,pst-node,subfigure,rotating, amsmath, bbm, amsthm, amssymb, amsthm, setspace, picture, epsfig, amsfonts, upgreek}

\usepackage{multirow}

\usepackage{arydshln}

\theoremstyle{definition}

\newtheorem{theorem}{Theorem}
\newtheorem{assumption}{Assumption}
\newtheorem{lemma}{Lemma}
\newtheorem{remark}{Remark}
\newtheorem{proposition}{Proposition}
\newtheorem{example}{Example}
\newtheorem{definition}{Definition}

\newcommand{\bm}[1]{\mbox{\boldmath$ #1 $\unboldmath}}

\def\P{{\mathrm{P}}}\def\v{\mathbf}\def\m{\boldsymbol}
\def\vb{\bm{\upbeta}}\def\ve{\bm{\upvarepsilon}}\def\vp{\bm{\upphi}}
\def\vt{\bm{\uptheta}}\def\va{\bm{\upalpha}}

\begin{document}

\begin{center}
    {\Large\bf OEM for least squares problems}
\\[2mm] Shifeng \textsc{Xiong}$^1$, Bin \textsc{Dai}$^{2}$ and Peter Z. G.
\textsc{Qian}\footnote[3]{Corresponding author: Peter Z. G. Qian. Email: peterq@stat.wisc.edu}
\\$^1$ Academy of Mathematics and Systems Science \\Chinese Academy of Sciences, Beijing 100190
\\$^2$ Tower Research Capital, 377 Broadway, New York, NY 10013
\\$^3$ Department of Statistics\\ University of Wisconsin-Madison, Madison, WI 53706
\end{center}

\vspace{3mm}

{\small\noindent{\bf Abstract}\quad We propose an algorithm, called OEM (a.k.a. orthogonalizing EM), intended for various least squares problems. The first step, named active orthogonization, orthogonalizes an arbitrary regression matrix by elaborately adding more rows. The second step imputes the responses of the new rows. The third step solves the least squares problem of interest for the complete orthogonal design. The second and third steps have simple closed forms, and iterate until convergence. The algorithm works for ordinary least squares and regularized least squares with the lasso, SCAD, MCP and other penalties. It has several attractive theoretical properties. For the ordinary least squares with a singular regression matrix, an OEM sequence converges to the Moore-Penrose generalized inverse-based least squares estimator. For the SCAD and MCP, an OEM sequence can achieve the oracle property after sufficient iterations for a fixed or diverging number of variables. For ordinary and regularized least squares with various penalties, an OEM sequence converges to a point having grouping coherence for fully aliased regression matrices. Convergence and convergence rate of the algorithm are examined. These convergence rate results show that for the same data set, OEM converges faster for regularized least squares than ordinary least squares. This provides a new theoretical comparison between these methods. Numerical examples are provided to illustrate the proposed algorithm.}

\vspace{1mm}

\noindent KEY WORDS: Design of experiments; MCP; Missing data; Optimization; Oracle property; Orthogonal design; SCAD; The Lasso.

\newpage
\section{INTRODUCTION}\label{sec:intro}
\hskip\parindent\vspace{-0.6cm}

Consider a regression model \begin{equation}\label{lm} \v{y}=\m{X}\vb+\ve, \end{equation}where $\m{X}=(x_{ij})$ is the $n\times p$ regression matrix,
$\v{y}\in{\mathbb{R}}^n$ is the response vector, $\vb=(\beta_1,\ldots,\beta_p)'$ is the vector of regression coefficients, and $\ve$ is the vector of random errors with
zero mean. The ordinary least squares (OLS) estimator of $\vb$ is the solution to\begin{equation}\label{ols}
\min_{{\scriptsize\vb}}\|\v{y}-\m{X}\vb\|^2,\end{equation}where $\|\cdot\|$ denotes the Euclidean norm. If $\m{X}$ is a part of a \emph{known} $m\times p$ orthogonal
matrix\begin{equation}\label{Xc} \m{X}_{\mathrm{c}}=\left(\begin{array}{c}\m{X}\\\m{\Delta}\end{array}\right),\end{equation}where $\m{\Delta}$ is an $(m-n)\times p$
matrix, \eqref{ols} can be efficiently computed by the Healy-Westmacott procedure (Healy and Westmacott 1956). Let\begin{equation}\label{Yc} \v{y}_{\mathrm{c}}=(\v{y}',\
\v{y}_{\mathrm{miss}}')'
\end{equation}be the vector of complete responses with missing data $\v{y}_{\mathrm{miss}}$ of $m-n$ points.
In each iteration, the procedure imputes the value of $\v{y}_{\mathrm{miss}}$, and updates the OLS estimator for the complete data $(\m{X}_{\mathrm{c}},
\v{y}_{\mathrm{c}})$. This update involves no matrix inversion since $\m{X}_{\mathrm{c}}$ is (column) orthogonal. Dempster, Laird, and Rubin (1977)
showed that this procedure is an EM algorithm.

The major limitation of the procedure is the assumption that $\m{X}$ must be embedded in a \emph{pre-specified} orthogonal matrix $\m{X}_{\mathrm{c}}$. We propose a new algorithm, called \emph{orthogonalizing EM} (OEM) algorithm, to remove this restriction and extend to other directions. The first step, called active
orthogonization, orthogonalizes an arbitrary regression matrix by elaborately adding more rows. The second step imputes the responses of the new rows. The third step solves the OLS problem in \eqref{ols} for the complete orthogonal design. The second and third steps have simple closed forms, and iterate until convergence.

For the OLS problem in (\ref{ols}), OEM works with an arbitrary regression matrix $\m{X}$. For $\m{X}$ with no full column rank, the OLS estimator is not unique, and we prove that the OEM
algorithm converges to the Moore-Penrose generalized inverse-based least squares estimator. OEM outperforms
existing methods for such an inverse.

OEM also works for \emph{regularized} least squares problems by adding penalties or constraints to $\vb$ in (\ref{ols}). These penalties include the ridge regression
(Hoerl and Kennard 1970), the nonnegative garrote (Breiman 1995), the lasso (Tibshirani 1996), the SCAD (Fan and Li 2001), and the MCP (Zhang 2010), among others. Here, the first step of OEM uses the same active orthogonalization as that for OLS. The second and third steps of OEM imputes the missing data and solves the regularized problem
for the complete data $(\m{X}_{\mathrm{c}},\v{y}_{\mathrm{c}})$. Both the second and third steps have a simple closed form. We prove that OEM converges to a local minimum or stable point of the regularized least squares problem under mild conditions. Convergence rate of OEM is also established. These convergence rate results show that for the same data, OEM converges faster for regularized least squares than ordinary least squares. This difference provides a new theoretical comparison between these methods. Compared with existing algorithms, OEM possesses two unique theoretical features. 1. \emph{Achieving the oracle property for nonconvex penalties}: An estimator of $\vb$ in (\ref{lm}) having the oracle property can not only select the correct submodel asymptotically, but also estimate the nonzero coefficients as efficiently as if the correct submodel were known in advance. Fan and Li (2001) proved that there exists a local solution of SCAD with this property. From the optimization viewpoint, the SCAD problem can have
\emph{many} local minima (Huo and Chen 2010) and it is not clear which one has this property. Zou and Li (2008) proposed the local linear approximation
(LLA) algorithm to solve the SCAD problem and showed that the one-step LLA estimator has the oracle property with a good initial estimator for a fixed $p$. The LLA estimator is not guaranteed to be a local minimum of SCAD. To the best of our knowledge, no theoretical results so far show that any existing algorithm can provide such a local minimum. We prove that the OEM solution for SCAD can achieve a local solution with this property. 2. \emph{Having
grouping coherence}: An estimator of $\vb$ is said to have grouping coherence if it has the same coefficient for full aliased columns in $\m{X}$. For the lasso, SCAD, and MCP, an OEM sequence converges to a point having grouping coherence, which implies that the full aliased variables will be in or out of the selected model together. This property cannot be achieved by existing algorithms including the coordinate descent algorithm. In terms of numerical performance, OEM can be very fast for ordinary least squares problems and SCAD for big tall data with $n > p$. For big wide data with $p > n$, OEM can be slow. This drawback can be mitigated by adopting a two-stage procedure like that in Fan and Lv (2008), where the first stage uses a screening approach to reduce the dimensionality to a moderate size, and the second stage uses OEM.

The remainder of the article will unfold as follows. Section~\ref{sec:arc} discusses the active orthogonalization procedure. Section \ref{sec:gi} presents OEM for OLS. Section \ref{sec:rls} extends OEM to regularized least squares. Section~\ref{sec:seqconv} provides convergence properties of OEM.
Section~\ref{sec:fulali} shows that for a regression matrix with full aliased columns, an OEM sequence for the lasso, SCAD, or MCP converges to a solution with grouping coherency. Section~\ref{sec:op} establishes the oracle property of the OEM solution for SCAD and MCP. Section \ref{sec:nc} presents numerical examples to compare OEM with other algorithms for regularized least squares. Section \ref{sec:diss} concludes with some discussion.

\section{ACTIVE ORTHOGONALIZATION}\label{sec:arc}
\hskip\parindent\vspace{-0.6cm}

For an arbitrary $n\times p$ matrix $\m{X}$ in \eqref{lm}, we propose \emph{active orthogolization} to actively orthogolize an arbitrary matrix by elaborately adding more rows. Let $\m{S}$ be a $p\times p$ diagonal matrix with non-zero diagonal elements $s_1,\ldots,s_p$. Define
\begin{equation}\label{zmat}
\m{Z}=\m{X}\m{S}^{-1}.
\end{equation}
Consider the eigenvalue decomposition $\m{V}'\m{\Gamma}\m{V}$ of $\m{Z}'\m{Z}$ (Wilkinson 1965), where $\m{V}$ is an orthogonal matrix and $\m{\Gamma}$ is a diagonal
matrix whose diagonal elements, $\gamma_1\geqslant\cdots\geqslant\gamma_p$, are the nonnegative eigenvalues of $\m{Z}'\m{Z}$. For $d\geqslant\gamma_1$, let
\begin{equation}\label{t}
t=\#\{j:\ \gamma_j=d,\ j=1,\ldots,p\}
\end{equation}
denote the number of the $\gamma_j$ equal $d$. For example, if $d=\gamma_1=\gamma_2$ and $\gamma_1>\gamma_j$ for $j=3,\ldots,p$, then $t=2$. If $d>\gamma_1$,
then $t=0$.
Define
\begin{equation}
\m{B}={\mathrm{diag}}(d -\gamma_{t+1},\ldots,d-\gamma_p)\label{B}
\end{equation} and
\begin{equation}\label{orth}
\m{\Delta}=\m{B}^{1/2}\m{V}_1\m{S},
\end{equation} where
$\m{V}_1$ is the submatrix of $\m{V}$ consisting of the last $p-t$ rows. Put $\m{X}$ and $\m{\Delta}$ row by row together to form a complete matrix $\m{X}_{\mathrm{c}}$.

\begin{lemma}\label{lemma}
The matrix $\m{X}_{\mathrm{c}}$ above is column orthogonal.
\end{lemma}

\begin{proof}
From (\ref{B}) and (\ref{orth}),
\begin{eqnarray*}
\m{X}_{\mathrm{c}}' \m{X}_{\mathrm{c}}=\m{X}' \m{X}+\m{\Delta}'\m{\Delta}=\m{S}(\m{V}'\m{\Gamma}\m{V}+\m{V}'_1\m{B}\m{V}_1)\m{S}.
\end{eqnarray*}For the $p\times p$ identity matrix $\m{I}_p$,$$d\m{I}_p-\m{\Gamma}=\left(\begin{array}{cc}\m{0}&\m{0}\\\m{0}&\m{B}\end{array}\right)$$
It then follows that $\m{X}_{\mathrm{c}}' \m{X}_{\mathrm{c}}=\m{S}[\m{V}'\m{\Gamma}\m{V}+\m{V}'(d\m{I}_p-\m{\Gamma})\m{V}]\m{S}=d\m{S}^2$, which completes the proof.
\end{proof}

Here is the underlying geometry of active orthogolization. For a vector $\v{x}\in{\mathbb{R}}^m$, let $P_\omega\v{x}$ denote its projection onto a subspace $\omega$ of
${\mathbb{R}}^m$. Lemma \ref{lemma} implies that for the column vectors $\v{x}_{1},\ldots,\v{x}_{p}\in{\mathbb{R}}^n$ of $\m{X}$ in (\ref{lm}), there exists a
set of mutually orthogonal vectors $\v{x}_{c1},\ldots,\v{x}_{cp}\in{\mathbb{R}}^{n+p-t}$ of $\m{X}_{\mathrm{c}}$ in (\ref{Xc}) satisfying
$P_{{\mathbb{R}}^n}\v{x}_{ci}=\v{x}_i$, for $j=1,\ldots,p$. Proposition 1 makes this precise.

\begin{proposition}
Let $\omega$ be an $n$-dimensional subspace of ${\mathbb{R}}^m$ with $n\leqslant m$. If $p\leqslant m-n+1$, then for any $p$ vectors $\v{x}_1,\ldots,\v{x}_p\in\omega$,
there exist $p$ vectors $\v{x}_{\mathrm{c}1},\ldots,\v{x}_{\mathrm{c}p}\in {\mathbb{R}}^m$ such that $P_\omega\v{x}_{\mathrm{c}i}=\v{x}_i$ for $j=1,\ldots,p$ and
$\v{x}_{\mathrm{c}i}'\v{x}_{\mathrm{c}j}=0$ for $i\neq j$.
\end{proposition}
\noindent For illustration, Figure \ref{fig:geo} expands two vectors $\v{x}_1$ and $\v{x}_2$ in $\mathbb{R}^2$ to two orthogonal vectors $\v{x}_{c1}$ and $\v{x}_{c2}$ in
$\mathbb{R}^3$.

\begin{figure}[t]\begin{center}\scalebox{0.8}[0.8]{\includegraphics{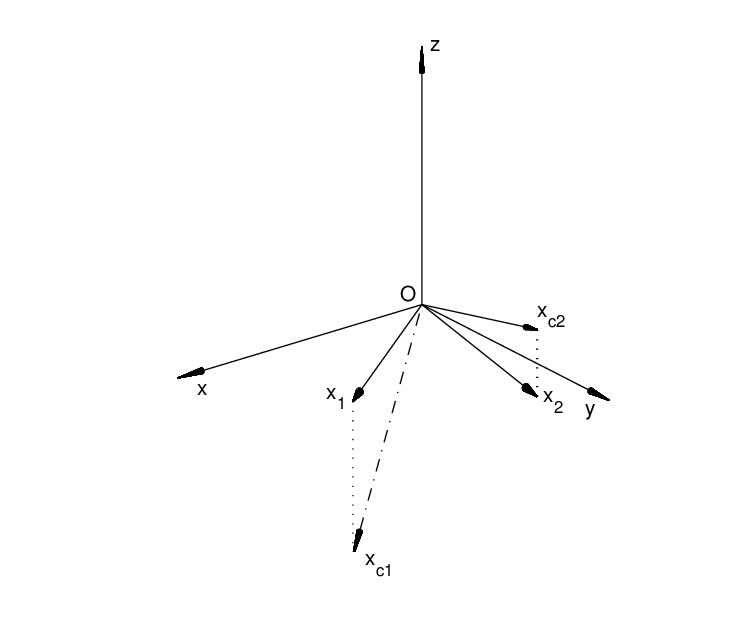}}\end{center}
\caption{Expand two two-dimensional vectors $\v{x}_1$ and $\v{x}_2$ to two three-dimensional vectors $\v{x}_{\mathrm{c}1}$ and $\v{x}_{\mathrm{c}2}$ with
$\v{x}_{\mathrm{c}1}'\v{x}_{\mathrm{c}2}=0$.}\label{fig:geo}
\end{figure}

\begin{remark}\label{rm:n}
In (\ref{orth}) $\m{\Delta}$ has $p-t$ rows, which does not rely on the number of rows in $\m{X}$, and only $p-t$ rows need to be added to make
it orthogonal.
\end{remark}

\begin{remark}\label{rm:1}
The form of $\m{S}$ in (\ref{zmat}) can be chosen flexibly. One possibility is $\m{S}=\m{I}_p$ with
\begin{equation}\m{X}'\m{X}+\m{\Delta}'\m{\Delta}=d\m{I}_p\label{XcdI}\end{equation}with $d\geqslant\gamma_1$, and
$\m{X}_{\mathrm{c}}$ is standardized as the Euclidean norm of each column is $d$.
\end{remark}

\begin{example} Suppose that $\m{X}$ in (\ref{lm}) is orthogonal. Take $d=\gamma_1$ and
\begin{equation}
\m{S}={\mathrm{diag}}\left[(\sum_{i=1}^nx_{i1}^2)^{1/2},\ldots,(\sum_{i=1}^nx_{ip}^2)^{1/2}\right].\label{SS}
\end{equation}
Since $t=p$, $\m{\Delta}$ in (\ref{orth}) is empty, which indicates that active orthogonalization will not overshoot.
\end{example}

\begin{example} Let $$\m{X}=\left(\begin{array}{rrrrr}0&\ &0&\ &3/2\\-4/3&&-2/3&&1/6\\2/3&&4/3&&1/6\\-2/3&&2/3&&-7/6\end{array}\right).$$
If $\m{S}=\m{I}_3$ and $d=\gamma_1$, (\ref{orth}) gives $\m{\Delta}=(-2/\sqrt{3},\ 2/\sqrt{3},\ 1/\sqrt{3}).$
\end{example}

\begin{example}\label{ex:design}
Consider a two-level design in three factors
$$\left(\begin{array}{rrr}-1&\ -1&\ -1\\-1&\ 1&\ 1\\1&\ -1&\ 1\\1&\ 1&\ -1\end{array}\right).$$
The regression matrix including all main effects and two-way interactions is
$$\m{X}=\left(\begin{array}{rrrrrr}-1&\ -1&\ -1&\ 1&\ 1&\ 1\\-1&\ 1&\ 1&-1&-1&1\\1&\ -1&\ 1&\ -1&\ 1&\ -1\\1&\ 1&\ -1&\ 1&\ -1&\ -1\end{array}\right),$$
where the last three columns for the interactions are fully aliased with the first three columns for the main effects. For $\m{S}=\m{I}_3$ and $d=\gamma_1$, (\ref{orth})
gives
$$\m{\Delta}=\left(\begin{array}{rrrrrr}0&\ -2&\ 0&\ 0&\ -2&\ 0\\0&\ 0&\ -2&\ -2&\ 0&\ 0\\-2&\ 0&\ 0&\ 0&\ 0&\ -2\end{array}\right).$$
The structure of $\m{\Delta}$ is flexible in that the interaction columns do not need to be a product of other two columns.

\end{example}

\begin{example} Consider a $1000\times 10$ random matrix $\m{X}=(x_{ij})$ with entries independently drawn  from
the uniform distribution on $[0,1)$. Using $\m{S}$ in (\ref{SS}), (\ref{orth}) gives
$$\m{\Delta}=\left(\begin{array}{rrrrrrrrrr}
-7.99&16.06&-6.39&-18.26&12.91&-8.67&7.56&34.08&-17.04&-11.81\\
26.83&-12.09&7.91&1.02&-22.75&-6.90&-19.98&26.10&-0.86&0.88\\
-4.01&1.48&9.51&-21.99&19.46&-10.27&-25.12&-3.39&7.29&27.90\\
21.77&10.72&-0.61&-6.46&28.00&1.28&-6.86&-7.04&11.13&-30.64\\
-15.78&5.60&-15.26&-7.67&-9.76&23.93&-14.71&12.25&29.45&-7.89\\
16.34&10.61&-41.82&11.82&6.49&-7.38&-6.14&-1.82&-1.86&13.09\\
-8.15&24.97&12.11&24.35&3.66&-2.59&-27.84&-3.45&-9.40&-13.72\\
-5.35&-21.70&-4.16&7.42&13.98&29.84&-10.26&7.60&-25.13&7.78\\
-19.62&-22.43&-2.61&22.58&11.80&-22.08&1.25&15.87&14.94&0.31\\
\end{array}\right).$$ Only nine rows need to be added to make this large $\m{X}$ matrix orthogonal. 
\end{example}







\section{OEM FOR ORDINARY LEAST SQUARES}\label{sec:gi}
\hskip\parindent\vspace{-0.6cm}

We now study OEM for the OLS problem in (\ref{ols}) when the regression matrix $\m{X}$ has an arbitrary form.
The first step of OEM is active orthogonalization to obtain $\m{\Delta}$ in \eqref{orth}. For an initial estimator $\vb^{(0)}$, the second step imputes
$\v{y}_{\mathrm{miss}}$ in (\ref{Yc}) by $\v{y}_{\mathrm{I}}=\m{\Delta}\vb^{(0)}$. Let $\v{y}_{\mathrm{c}}=(\v{y}',\ \v{y}_{\mathrm{I}}')'$. The third step solves
\begin{eqnarray}\vb^{(1)}=\arg\!\min_{{\scriptsize\vb}} \|\v{y}_{\mathrm{c}}-\m{X}_{\mathrm{c}}\vb\|^2.\label{iols0}\end{eqnarray}Then, the second and third
steps iterate for obtaining $\vb^{(2)},\vb^{(3)},\ldots$
until convergence. Define
\begin{equation}\label{Ad}
\m{A}=\m{\Delta}'\m{\Delta}.
\end{equation}
For $\m{X}_{\mathrm{c}}$ in \eqref{Xc}, let $(d_1,\ldots,d_p)$ denote the diagonal elements of $\m{X}_{\mathrm{c}}' \m{X}_{\mathrm{c}}$.
For $k=0,1,\ldots,$ let
\begin{equation}\label{vu}
\v{u}=(u_1,\ldots,u_p)'=\m{X}' \v{y}+\m{A}\vb^{(k)},
\end{equation}and
(\ref{iols0}) becomes
\begin{equation}\label{iols}
\vb^{(k+1)}=\arg\!\min_{{\scriptsize\vb}}\sum_{j=1}^p(d_j\beta_j^2-2u_j\beta_j),
\end{equation}
which is \emph{separable} in the dimensions of $\vb$. Thus, (\ref{iols}) has a simple form
\begin{equation}\label{olss}
\beta_j^{(k+1)}=u_j/d_j,\quad\text{for}\ j=1,\ldots,p,\end{equation} which involves no matrix inversion.

In (\ref{vu}), for active orthogonalization, instead of computing $\m{\Delta}$ in (\ref{orth}), one can compute $\m{A}=\m{\Delta}' \m{\Delta}$ and the diagonal entries
$d_1,\ldots,d_p$ of $\m{X}_{c}' \m{X}_{\mathrm{c}}$. If $\m{S}=\m{I}_p$ in (\ref{orth}), $\m{A}=d\m{I}_p-\m{X}' \m{X}$, where $d$ is a number no less than
the largest eigenvalue $\gamma_1$ of $\m{X}'\m{X}$. A possible choice is $d={\mathrm{trace}}(\m{X}'\m{X})$. Another choice is $d=\gamma_1$ to obtain the fastest
convergence; see Remark \ref{rm:r}. We compute $\gamma_1$ by the power method (Wilkinson 1965) described below. Given a nonzero initial vector
$\v{a}^{(0)}\in{\mathbb{R}}^p$, let $\gamma_1^{(0)}=\|\v{a}^{(0)}\|$. For $k=0,1,...$, compute
$\v{a}^{(k+1)}=\m{X}'\m{X}\v{a}^{(k)}/\gamma_1^{(k)}$ and $\gamma_1^{(k+1)}=\|\v{a}^{(k+1)}\|$
until convergence. If $\v{a}^{(0)}$ is not an eigenvector of any $\gamma_j$ unequal to $\gamma_1$, then $\gamma_1^{(k)}$ converges to $\gamma_1$.
For $t$ in (\ref{t}), the convergence rate of the power method is linear (Watkins 2002) specified by
$$\lim_{k\rightarrow\infty}\frac{\|\gamma_1^{(k+1)}-\gamma_{1}\|}{\|\gamma_1^{(k)}-\gamma_{1}\|}=\frac{\gamma_{t+1}}{\gamma_1}.$$When $p>n$,
replace the $p\times p$ matrix $\m{X}'\m{X}$ with the $n\times n$ matrix $\m{X}\m{X}'$ in the power method to reduce computational cost as the two matrices have the same non-zero eigenvalues.

When $\m{X}$ has full column rank, the convergence results in Wu (1983) indicates that the OEM sequence given by (\ref{olss}) converges to the OLS
estimator for any initial point $\vb^{(0)}$. Next, we discuss the convergence property of OEM when $\m{X}'\m{X}$ is singular, which covers the case
of $p>n$. Let $r$ denote the rank of $\m{X}$. For $r<p$, the singular value decomposition (Wilkinson 1965) of $\m{X}$ is
$$\m{X}=\m{U}'\left(\begin{array}{cc}\m{\Gamma}_0^{1/2}&\m{0}\\\m{0}&\m{0}\end{array}\right)\m{V},$$where $\m{U}$ is an $n\times n$ orthogonal matrix, $\m{V}$ is a $p\times
p$ orthogonal matrix, and $\m{\Gamma}_0$ is a diagonal matrix with diagonal elements $\gamma_1\geqslant\cdots\geqslant\gamma_r$ which are the positive eigenvalues of
$\m{X}' \m{X}$. Define\begin{equation}\hat{\vb}^*=(\m{X}'\m{X})^+\m{X}'\v{y},\label{vb}\end{equation} where $^+$ denotes the Moore-Penrose generalized inverse
(Ben-Israel and Greville 2003).

\begin{theorem}Suppose that $\m{X}' \m{X}+\m{\Delta}'\m{\Delta}=\gamma_1\m{I}_p$. If $\vb^{(0)}$ lies in the linear space spanned by the first $r$ columns
of $\m{V}'$, then as $k\rightarrow\infty$, for the OEM sequence $\{\vb^{(k)}\}$ of the ordinary least squares,
$\vb^{(k)}\rightarrow\hat{\vb}^*$.\label{th:olsconve}\end{theorem}
\begin{proof} Define
$\m{D}=\m{I}_p-\gamma_1^{-1}\m{X}'\m{X}$. Note that $\vb^{(k+1)}=\gamma_1^{-1}\m{X}'\v{y}+\m{D}\vb^{(k)}$. By induction,
\begin{eqnarray*}\vb^{(k)}&=&\gamma_1^{-1}(\m{I}_p+\m{D}+\cdots+\m{D}^{k-1})\m{X}'\v{y}+\m{D}^k\vb^{(0)}
\\&=&\gamma_1^{-1}\m{V}'\left\{\m{I}_p+\left(\begin{array}{cc}\m{I}_r-\gamma_1^{-1}\m{\Gamma}_0&\m{0}\\\m{0}&-\m{I}_{p-r}\end{array}\right)+\cdots
+\left(\begin{array}{cc}(\m{I}_r-\gamma_1^{-1}\m{\Gamma}_0)^{k-1}&\m{0}\\\m{0}&(-1)^{k-1}\m{I}_{p-r}\end{array}\right)\right\}\\&&\ \cdot\m{V}
\m{V}'\left(\begin{array}{cc}\m{\Gamma}_0^{1/2}&\m{0}\\\m{0}&\m{0}\end{array}\right)\m{U}\v{y}+\m{D}^k\vb^{(0)}
\\&=&\gamma_1^{-1}\m{V}'\left(\begin{array}{cc}\big\{\m{I}_r+(\m{I}_r-\gamma_1^{-1}\m{\Gamma}_0)+\cdots+(\m{I}_r-\gamma_1^{-1}\m{\Gamma}_0)^{k-1}\big\}
\m{\Gamma}_0^{1/2}&\m{0}\\\m{0}&\m{0}
\end{array}\right)\m{U}\v{y}+\m{D}^k\vb^{(0)}.\end{eqnarray*} As $k\rightarrow\infty$,
$$\m{D}^k\rightarrow \m{V}'\left(\begin{array}{cc}\m{0}&\\&\m{I}_{p-r}\end{array}\right)\m{V}$$and
$\m{D}^k\vb^{(0)}\rightarrow0$, which implies that $$\vb^{(k)}\rightarrow \m{V}'\left(\begin{array}{cc}\m{\Gamma}_0^{-1/2}&\m{0}\\\m{0}&\m{0}\end{array}\right)
\m{U}\v{y}=\hat{\vb}^*.$$
This completes the proof.\end{proof}

In active orthogonalization, the condition  $\m{X}' \m{X}+\m{\Delta}'\m{\Delta}=\gamma_1\m{I}_p$ holds if $d=\gamma_1$ and $\m{S}=\m{I}_p$ in (\ref{orth}). Using
$\vb^{(0)}=\v{0}$ satisfies the condition in Theorem \ref{th:olsconve}.

The Moore-Penrose generalized inverse is widely used in statistics for a degenerated system. Theorem \ref{th:olsconve} indicates that OEM converges to $\hat{\vb}^*$ in
\eqref{vb} in this case. When $r<p$, the limiting vector $\hat{\vb}^*$ given by an OEM sequence has the following properties. First, it has the minimal Euclidean norm
among the least squares estimators $(\m{X}'\m{X})^-\m{X}'\v{y}$ (Ben-Israel and Greville 2003). Second, its model error has a simple form,
$E\big[(\hat{\vb}^*-\vb)'(\m{X}' \m{X})(\hat{\vb}^*-\vb)\big]=r\sigma^2$. Third, $\m{X}\va=\v{0}$ implies $\va'\hat{\vb}^*=0$ for any vector $\va$. The third property
indicates that $\hat{\vb}^*$ inherits the multicollinearity between the columns in $\m{X}$. This property is stronger than grouping coherence for regularized least squares in
Section \ref{sec:fulali}.

A widely used method for computing $\hat{\vb}^*$ is to obtain $(\m{X}'\m{X})^+$ by the eigenvalue decomposition (Golub and Van Loan 1996) and then compute the
product of $(\m{X}'\m{X})^+$ and $\m{X}'\v{y}$. This method is implemented in the \texttt{MATLAB} function
\texttt{pinv} with core code in \texttt{Fortran} and in the \texttt{R} function \texttt{ginv} with core code in \texttt{C++}. The \texttt{R} function is slower than the \texttt{MATLAB} function. When some eigenvalues are close to zero, the eigenvalue decomposition method is unstable, and OEM is more stable due to its iterative nature. The following example illustrates this difference.

\begin{example} Construct a $10\times4$ matrix $\m{X}=\big({\mathrm{diag}}(1,1,1,\sqrt{u})\ \m{0}\big)'$, where $u$ is generated from a uniform
distribution on $[10^{-16},\ 10^{-14})$. The eigenvalues of $\m{X}'\m{X}$ are $1,1,1$, and $u$. Generate all entries of $\v{y}$ independently from the uniform
distribution on $[0,1)$. We compare OEM and the eigenvalue decomposition method for computing $\hat{\vb}^*$ using the \texttt{MATLAB} function \texttt{pinv}. For OEM, the stopping criterion is
when relative changes in all coefficients are less than $10^{-4}$. The two methods are replicated 100 times in \texttt{MATLAB}. Over the 100 replicates, the largest and
smallest values of $\|\hat{\vb}^*\|$ by the eigenvalue decomposition method are $2.06\times10^{7}$ and $0.25$, indicating unstability. The two values computed by OEM are $1.48$ and $0.25$,
which are much more stable.
\end{example}

Next, we discuss the computational efficiency of OEM for computing $\hat{\vb}^*$ in \eqref{vb} when $\m{X}$ is degenerated. Recall that $\m{X}'\m{X}$ and
$\m{X}\m{X}'$ have the same nonzero eigenvalues. The computation of $\gamma_1$ in the OEM iterations by the power method has complexity
$O\big(\min\{n,p\}^2\max\{n,p\}\big)$. Since the complexity of the OEM iterations is $O(np^2)$, the whole computational complexity of OEM for computing $\hat{\vb}^*$ is
$O(np^2)$. The eigenvalue decomposition method computes $(\m{X}'\m{X})^+$ first by eigenvalue decomposition to obtain $\hat{\vb}^*$, and has computational complexity $O(np^2+p^3)$. The
OEM algorithm is superior to this method in terms of complexity.

\begin{table}[htb]
\begin{center}
\begin{tabular}{cc|ccc}
$n$& $p$ & OEM & eigenvalue \\ & & &decomposition\\
\hline
\hline \multirow{5}{*}{$50,000$}
& $10$ & $0.0433$ & $0.0956$ \\
&$50$ & $0.2439$ & $0.4098$ &\\
&$200$ & $1.4156$ & $4.9765$\\
&$1000$ & $5.4165$ & $45.3270$ \\
&$5,000$ & $72.0630$ & $442.3300$
\end{tabular}
\caption{Average runtime (second) comparison between OEM and the prevailing method for $n > p$} \label{comp3}
\end{center}
\end{table}

\begin{table}[htb]
\begin{center}
\begin{tabular}{cc|cc}
$p$& $n$ & OEM & eigenvalue \\ & & &decomposition\\
\hline
\hline \multirow{5}{*}{$50,000$}
& $10$ & $0.0482$ & $0.1153$ \\
&$50$ & $0.4203$ & $0.4176$ \\
&$200$ & $1.9159$ & $5.2053$\\
&$1000$ & $8.4626$ & $47.7653$ \\
&$5,000$ & $71.8477$ & $440.6741$
\end{tabular}
\caption{Average runtime (second) comparison between OEM and the eigenvalue decomposition method for $p>n$} \label{comp4}
\end{center}
\end{table}

We conduct a simulation study to compare the speeds of OEM and the eigenvalue decomposition method for computing $\hat{\vb}^*$ in \eqref{vb}. Generate all entries of $\m{X}$ and $\v{y}$ independently from the standard normal distribution. A new predictor calculated as the mean of all the covariates is added to degenerate the design matrix. Tables \ref{comp3} and
\ref{comp4} compare our \texttt{R} package \texttt{oem} with main code in \texttt{C++} and the eigenvalue decomposition method in computing $\hat{\vb}^*$.
The two methods give the same results. Tables \ref{comp3} and \ref{comp4} indicate that OEM is faster than the eigenvalue decomposition method for any combination
of $n$ and $p$, validating the above complexity analysis.

\section{OEM FOR REGULARIZED LEAST SQUARES}\label{sec:rls}
\hskip\parindent\vspace{-0.6cm}

It is easy to extend OEM to regularized least squares problems. Consider a penalized version of (\ref{lm}):
\begin{equation}\label{pls}
\min_{{\scriptsize\vb}\in\Theta}\left[ \|\v{y}-\m{X}\vb\|^2+P(\vb;\lambda) \right],
\end{equation} where ${{\vb}\in\Theta}$,
$\Theta$ is a subset of ${\mathbb{R}}^p$, $P$ is a penalty function, and $\lambda$ is the vector of tuning parameters. To apply the penalty $P$ equally to all the
variables, the regression matrix $\m{X}$ is standardized so that
\begin{equation}
\sum_{i=1}^nx_{ij}^2=1, \;  \mbox{for} \; j=1,\ldots,p.\label{stand}
\end{equation} Popular choices for $P$ include the
ridge regression (Hoerl and Kennard 1970), the nonnegative garrote (Breiman 1995), the lasso (Tibshirani 1996), the SCAD (Fan and Li 2001), and the MCP (Zhang 2010).

Suppose that $\Theta$ and $P$ in (\ref{pls}) are \emph{decomposable} as $\Theta=\prod_{j=1}^p\Theta_j$ and $P(\vb;\lambda)=\sum_{j=1}^pP_j(\beta_j;\lambda)$. For the
problem in (\ref{pls}), the first step of OEM is active orthogonalization, which computes $\m{\Delta}$ in \eqref{orth}. For an initial estimator $\vb^{(0)}$, the second
step imputes $\v{y}_{\mathrm{miss}}$ in (\ref{Yc}) by $\v{y}_{\mathrm{I}}=\m{\Delta}\vb^{(0)}$. Let $\v{y}_{\mathrm{c}}=(\v{y}',\ \v{y}_{\mathrm{I}}')'$. The third step
solves
\begin{eqnarray*}\vb^{(1)}=\arg\!\min_{{\scriptsize\vb}\in\Theta}\left[ \|\v{y}_{\mathrm{c}}-\m{X}_{\mathrm{c}}\vb\|^2+P(\vb;\lambda) \right]. \end{eqnarray*}
The second and third steps iterate to compute $\vb^{(k)}$ for $k=1,2,\ldots$ until convergence. Similar to (\ref{iols}), we have an iterative formula
\begin{equation}\label{cw}
\beta_j^{(k+1)}=\arg\!\!\min_{\beta_j\in\Theta_j}\big[d_j\beta_j^2-2u_j\beta_j+P_j(\beta_j;\lambda)\big],\ \text{for}\
j=1,\ldots,p,
\end{equation}
with $\v{u}=(u_1,\ldots,u_p)'$ in (\ref{vu}). This shortcut applies to the following penalties:

\begin{description}

\item[]1. The lasso (Tibshirani 1996), where $\Theta_j={\mathbb{R}}$, \begin{equation}P_j(\beta_j;\lambda)=2\lambda|\beta_j|,\label{lassoP}\end{equation}
and (\ref{cw}) becomes
\begin{equation}\label{lasso}
\beta_j^{(k+1)}={\mathrm{sign}}(u_j)\left(\frac{|u_j|-\lambda}{d_j}\right)_+.
\end{equation}
Here, for $a \in \mathbb{R}$, $(a)_+$ denotes $\max\{a,0\}$.

\item[]2. The nonnegative garrote (Breiman 1995), where $\Theta_j=\{x:x\hat{\beta}_j\geqslant0\}$, $P_j(\beta_j;\lambda)=2\lambda\beta_j/\hat{\beta}_j$,
$\hat{\beta}_j$ is the OLS estimator of $\beta_j$, and (\ref{cw}) becomes
$$\beta_j^{(k+1)}=\left(\frac{u_j\hat{\beta}_j-\lambda}{d_j\hat{\beta}_j^2}\right)_+\hat{\beta}_j.$$

\item[]3. The elastic-net (Zou and Hastie 2005), where $\Theta_j={\mathbb{R}}$,
\begin{equation}P_j(\beta_j;\lambda)=2\lambda_1|\beta_j|+\lambda_2\beta_j^2.\label{netP}\end{equation}
and (\ref{cw}) becomes \begin{equation}\beta_j^{(k+1)}= {\mathrm{sign}}(u_j)\left(\frac{|u_j|-\lambda_1}{d_j+\lambda_2}\right)_+.\label{net}\end{equation}

\item[]5. The SCAD (Fan and Li 2001), where $\Theta_j={\mathbb{R}}$, $P_j(\beta_j;\lambda)=2P_\lambda(|\beta_j|)$, and \begin{equation}
\label{scadP} P_\lambda'(\theta)=\lambda I(\theta\leqslant\lambda)+(a\lambda-\theta)_+ I(\theta>\lambda)/(a-1),
\end{equation}with $a>2$, $\lambda\geqslant0$, and $\theta>0$. Here, $I$ is the indicator function. If $\m{X}$ in (\ref{lm})
is standardized as in (\ref{stand}) with $d_j\geqslant1$
for all $j$, (\ref{cw}) becomes \begin{equation}\label{scad} \beta_j^{(k+1)} =\left\{\begin{array}{ll}{\mathrm{sign}}(u_j)\big(|u_j|-\lambda\big)_+/d_j,&\text{when}\
|u_j|\leqslant(d_j+1)\lambda,
\\{\mathrm{sign}}(u_j)\big[(a-1)|u_j|-a\lambda\big]/\big[(a-1)d_j-1\big],\quad&\text{when}\ (d_j+1)\lambda<|u_j|\leqslant a\lambda d_j,
\\ u_j/d_j,&\text{when}\ |u_j|>a\lambda d_j.\end{array}\right.
\end{equation}

\item[]6. The MCP (Zhang 2010), where $\Theta_j={\mathbb{R}}$, $P_j(\beta_j;\lambda)=2P_\lambda(|\beta_j|)$, and \begin{equation}P_\lambda'(\theta)=(\lambda-\theta/a)
I(\theta\leqslant a\lambda)\label{mcpP}\end{equation} with $a>1$ and $\theta>0$. If $\m{X}$ in (\ref{lm}) is standardized as in (\ref{stand}) with $d_j\geqslant1$ for
all $j$, (\ref{cw}) becomes \begin{equation}\beta_j^{(k+1)} =\left\{\begin{array}{ll}{\mathrm{sign}}(u_j)a\big(|u_j|-\lambda\big)_+/(ad_j-1),\quad&\text{when}\
|u_j|\leqslant a\lambda d_j,
\\ u_j/d_j,&\text{when}\ |u_j|>a\lambda d_j.\end{array}\right.\label{mcp}\end{equation}

\item[]7. The ``Berhu" penalty (Owen 2006), where $\Theta_j={\mathbb{R}}$, $P_j(\beta_j;\lambda)=2\lambda\big\{|\beta_j|I(|\beta_j|<\delta)
+(\beta_j^2+\delta^2)I(|\beta_j| \geqslant\delta)/(2\delta)\big\}$ for some $\delta>0$, and (\ref{cw}) becomes
$$\beta_j^{(k+1)}=\left\{\begin{array}{ll}{\mathrm{sign}}(u_j)\big(|u_j|- \lambda\big)_+/d_j,\quad&\text{when}\ |u_j|<\lambda+d_j\delta,
\\ u_j\delta/(\lambda+d_j\delta),&\text{when}\ |u_j|\geqslant\lambda+d_j\delta.\end{array}\right.$$

\end{description}

OEM for (\ref{pls}) is an \textsc{EM} algorithm. Let the observed data $\v{y}$ follow the model in \eqref{lm}. Assume that the complete data $\v{y}_{\mathrm{c}}=(\v{y}',\ \v{y}_{\mathrm{miss}}')'$ in (\ref{Yc}) follows a regression model
$\v{y}_{\mathrm{c}}=\m{X}_{\mathrm{c}}\vb+\ve_c$, where $\ve_c$ is from $N(\v{0},\ \m{I}_m)$. Let $\hat{\vb}$ be a solution to (\ref{pls}) given by
$\hat{\vb}=\arg\!\max_{{\scriptsize\vb\in\Theta}}L(\vb\mid \v{y})$, and the regularized likelihood function $L(\vb\mid \v{y})$ is
$$(2\pi)^{-n/2}\exp\left(-\frac{1}{2}\|\v{y}-\m{X}\vb\|^2\right)\exp\left[-\frac{1}{2}P(\vb;\lambda)\right].$$
Given $\vb^{(k)}$, the second step of OEM for (\ref{pls}) is the E-step,
\begin{eqnarray*}
&&E\big[\log\{L(\vb|\v{y}_{\mathrm{c}})\}\mid \v{y},\vb^{(k)}\big]
\\&=&-C\big\{\|\v{y}-\m{X}\vb\|^2+E\big(\|\v{y}_{\mathrm{miss}}-\m{X}\vb\|^2\mid \vb^{(k)}\big)+P(\vb;\lambda)\big\}
\\&=&-C\big\{n+\|\v{y}-\m{X}\vb\|^2+\|\m{\Delta}\vb^{(k)}-\m{\Delta}\vb\|^2+P(\vb;\lambda)\big\}
\end{eqnarray*}
for some constant $C>0$. Define
\begin{equation}\label{Q}
Q(\vb\mid \vb^{(k)})=\|\v{y}-\m{X}\vb\|^2+\|\m{\Delta}\vb^{(k)}-\m{\Delta}\vb\|^2+P(\vb;\lambda).
\end{equation}
The third step of OEM is the M-step,
\begin{equation}\label{mQ}
\vb^{(k+1)}=\arg\min_{{\scriptsize\vb\in\Theta}}Q(\vb\mid \vb^{(k)}),
\end{equation} which is equivalent to (\ref{cw}) when $\Theta$ and $P$ in (\ref{pls}) are decomposable.

\begin{example}\label{ex:ffd}
For the model in (\ref{lm}), let the complete matrix $\m{X}_{\mathrm{c}}$ be an orthogonal design from Xu (2009) with 4096 runs in 30 factors. Let $\m{X}$ in (\ref{lm})
be the submatrix of $\m{X}_{\mathrm{c}}$ consisting of the first 3000 rows and let $\v{y}$ be generated from \eqref{lm} with $\sigma=1$ and
\begin{equation}\label{beta}
\beta_j=(-1)^j\exp\big[-2(j-1)/20\big]\ \text{for}\ j=1,\ldots,p.
\end{equation}
Here, let $p=30$, $n=3000$, and the response values for the last 1096 rows of $\m{X}_{\mathrm{c}}$ be missing. OEM is used to solve the SCAD problem with an initial
value $\vb^{(0)}=\v{0}$ and a stopping criterion when relative changes in all coefficients are less than $10^{-6}$. For $\lambda=1$ and $a=3.7$ in (\ref{scadP}), Figure
\ref{fig:objplot} plots values of the objective function in (\ref{pls}) with the SCAD penalty of the OEM sequence against iteration numbers, where the convergence occurs
at iteration 13, and the objective function significantly reduces after two iterations.
\end{example}

\begin{figure}[t]
\begin{center}\scalebox{0.5}[0.5]{\includegraphics{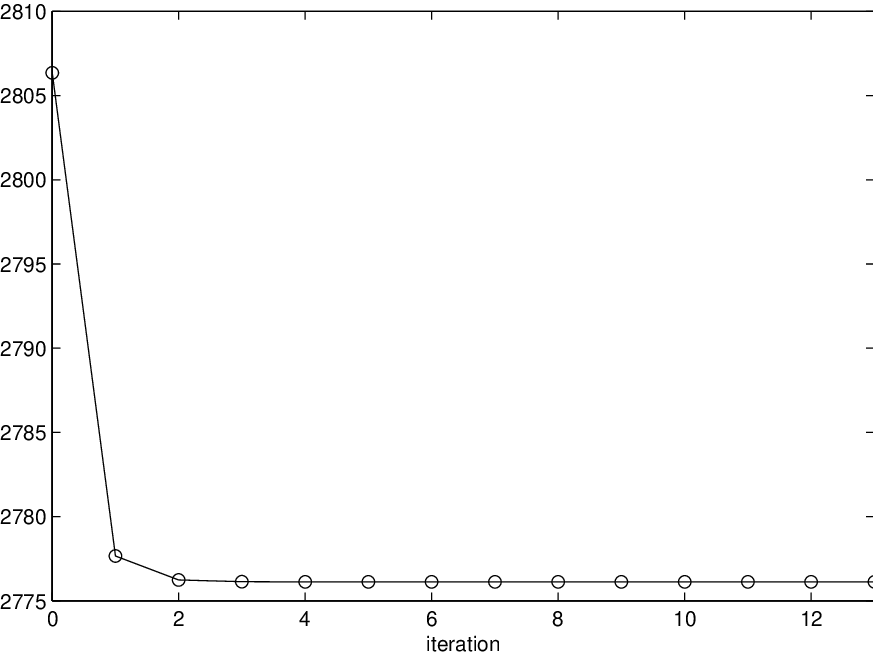}}\end{center}
\caption{Values of the objective function of an OEM sequence for the SCAD against iterations for Example \ref{ex:ffd}.}\label{fig:objplot}
\end{figure}

\section{CONVERGENCE OF THE OEM ALGORITHM}\label{sec:seqconv}
\hskip\parindent\vspace{-0.6cm}

We now derive convergence properties of OEM with the general penalty in (\ref{pls}). We also give results to compare the convergence rates of
OEM for OLS, the elastic-net, and the lasso. These convergence rate results show that for the same data set, OEM converges faster for regularized least squares than ordinary least squares. This provides a new theoretical comparison between these methods. The objective functions of existing EM convergence results like those in Wu (1983), Green (1990) and McLachlan and Krishnan
(2008) are typically continuously differentiable. This condition does not hold for the objective function in (\ref{pls}) with the lasso and other penalties, and these existing results do not directly apply here.

We make several assumptions for $\Theta$ and $P(\vb;\lambda)$ in (\ref{pls}).

\begin{assumption}
The parameter space $\Theta$ is a closed convex subset of ${\mathbb{R}}^p$.\label{as:theta}
\end{assumption}
\begin{assumption}For a fixed $\lambda$, the penalty $P(\vb;\lambda)\rightarrow+\infty$ as $\|\vb\|\rightarrow+\infty$.\label{as:P}\end{assumption}
\begin{assumption}For a fixed $\lambda$, the penalty $P(\vb;\lambda)$ is continuous with respect to $\vb \in \Theta$.\label{as:Pc}\end{assumption}
\noindent All penalties discussed in Section\ \ref{sec:rls} satisfy these assumptions. The assumptions cover the case in which the iterative sequence $\{\vb^{(k)}\}$
defined in (\ref{mQ}) may fall on the boundary of $\Theta$ (Nettleton 1999), like the nonnegative garrote (Breiman 1995) and the nonnegative lasso (Efron et~al. 2004).
The bridge penalty (Frank and Friedman 1993) in (\ref{bridge}) also satisfies the above assumptions.

For the model in (\ref{lm}), denote the objective function in (\ref{pls}) by
\begin{eqnarray}l(\vb)=\|\v{y}-\m{X}\vb\|^2+P(\vb;\lambda).\label{l}\end{eqnarray}
For penalties like the bridge, it is infeasible to perform the M-step in (\ref{mQ}) directly. For this situation, following the generalized EM algorithm in
Dempster, Laird, and Rubin (1977), we define the following \emph{generalized OEM} algorithm
\begin{eqnarray}\vb^{(k)}\rightarrow\vb^{(k+1)}\in\mathcal{M}(\vb^{(k)}),\label{goem}\end{eqnarray} where $\vb\rightarrow\mathcal{M}(\vb)\subset\Theta$ is a
point-to-set map such that $$Q(\vp\mid \vb)\leqslant Q(\vb\mid \vb),\quad\text{for all}\ \vp\in\mathcal{M}(\vb).$$ Here, $Q$ is given in (\ref{Q}). The OEM sequence
defined by (\ref{mQ}) is a special case of (\ref{goem}). For example, the generalized OEM algorithm can be used for the bridge penalty, where $\Theta_j={\mathbb{R}}$ and
\begin{equation}\label{bridge}
P_j(\beta_j;\lambda)=\lambda|\beta_j|^a
\end{equation}
for some $a\in(0,1)$ in (\ref{pls}). Since the solution to (\ref{cw}) with the bridge penalty has no closed form, one may use one-dimensional search to compute
$\beta_j^{(k+1)}$ that satisfies (\ref{goem}). By Assumption 1, $\{\vb\in\Theta: l(\vb)\leqslant l(\vb^{(0)})\}$ is compact for any $l(\vb^{(0)})>-\infty$. By Assumption
\ref{as:Pc}, $\mathcal{M}$ is a closed point-to-set map (Zangwill 1969; Wu 1983).

The objective functions in (\ref{pls}) with the lasso and other penalties are not continuously differentiable. A more general definition of stationary points is needed. We call $\vb\in\Theta$ a stationary point of $l$ if
$$\liminf_{t\rightarrow 0_+}\frac{l\big((1-t)\vb+t\vp\big)-l(\vb)}{t}\geqslant0\quad\text{for all}\ \vp\in\Theta.$$Let $S$ denote the set of stationary points of
$l$. Analogous to Theorem 1 in Wu (1983) on the global convergence of the EM algorithm, we have the following result.

\begin{theorem}\label{th:general}
Let $\{\vb^{(k)}\}$ be a generalized OEM sequence generated by (\ref{goem}). Suppose that
\begin{equation}
l(\vb^{(k+1)})<l(\vb^{(k)})\quad\text{for all }\ \vb^{(k)}\in\Theta\setminus S.\label{con}
\end{equation}
Then all limit points of $\{\vb^{(k)}\}$ are elements of $S$ and $l(\vb^{(k)})$ converges monotonically to $l^*=l(\vb^*)$ for some $\vb^*\in S$.
\end{theorem}

\begin{theorem}If $\vb^*$ is a local minimum of $Q(\vb\mid \vb^*)$, then $\vb^*\in S$.\label{th:local}
\end{theorem}
This theorem follows from the fact that $l(\vb)-Q(\vb\mid \vb^*)$ is differentiable and
$$\frac{\partial \big[l(\vb)-Q(\vb\mid \vb^*)\big]}{\partial \vb}\Big|_{\vb=\vb^*}=0.$$

\begin{remark}
By Theorem~\ref{th:local}, if $\vb^{(k)}\notin S$, then $\vb^{(k)}$ cannot be a local minimum of $Q(\vb\mid \vb^{(k)})$. Thus, there exists at least one point
$\vb^{(k+1)}\in \mathcal{M}(\vb^{(k)})$ such that $Q(\vb^{(k+1)}\mid \vb^{(k)})<Q(\vb^{(k)}\mid \vb^{(k)})$ and therefore satisfies the condition in (\ref{con}). As a
special case, an OEM sequence generated by (\ref{mQ}) satisfies (\ref{con}) in Theorem~\ref{th:general}.\end{remark}

Next, we derive convergence results of a generalized OEM sequence $\{\vb^{(k)}\}$ in (\ref{goem}), which, by Theorem \ref{th:local}, hold automatically for an
OEM sequence. If the penalty function $P(\vb;\lambda)$ is convex and $l(\vb)$ has a unique minimum, Theorem \ref{th:convex} shows that $\{\vb^{(k)}\}$ converges
to the global minimum.
\begin{theorem}For $\{\vb^{(k)}\}$ defined in Theorem \ref{th:general}, suppose that $l(\vb)$ in (\ref{l}) is a convex function on $\Theta$ with a unique
minimum $\vb^*$ and that (\ref{con}) holds for $\{\vb^{(k)}\}$. Then $\vb^{(k)}\rightarrow\vb^*$ as $k\rightarrow\infty$.\label{th:convex}\end{theorem}
\begin{proof} It suffices to show that $S=\{\vb^*\}$. For $\vp\in\Theta$ with $\vp\neq\vb^*$ and $t>0$,
$$\frac{l\big((1-t)\vp+t\vb^*\big)-l(\vb^*)}{t}\leqslant\frac{tl(\vb^*)+(1-t)l(\vp)-l(\vp)}{t} =l(\vb^*)-l(\vp)<0.$$This implies $\vp\notin S$.
\end{proof}

Theorem~\ref{th:conv} discusses the convergence of an OEM sequence $\{\vb^{(k)}\}$ for more general penalties. For $a\in{\mathbb{R}}$, define $S(a)=\{\vp\in S:\
l(\vp)=a\}$. From Theorem \ref{th:general}, all limit points of an OEM sequence are in $S(l^*)$, where $l^*$ is the limit of $l(\vb^{(k)})$ in Theorem \ref{th:general}.
Theorem \ref{th:conv} states that the limit point is unique under certain conditions.

\begin{theorem}Let $\{\vb^{(k)}\}$ be a generalized OEM sequence generated by (\ref{goem}) with $\m{\Delta}' \m{\Delta}>0$. If (\ref{con}) holds, then all
limit points of $\{\vb^{(k)}\}$ are in a connected and compact subset of $S(l^*)$. In particular, if the set $S(l^*)$ is discrete in that its only connected components
are singletons, then $\vb^{(k)}$ converges to some $\vb^*$ in $S(l^*)$ as $k\rightarrow\infty$.\label{th:conv}\end{theorem}
\begin{proof} Note that $Q(\vb^{(k+1)}\mid
\vb^{(k)})=l(\vb^{(k+1)})+\|\m{\Delta}\vb^{(k+1)}-\m{\Delta}\vb^{(k)}\|^2\leqslant Q(\vb^{(k)}\mid \vb^{(k)})=l(\vb^{(k)})$. By Theorem \ref{th:general},
$\|\m{\Delta}\vb^{(k+1)}-\m{\Delta}\vb^{(k)}\|^2\leqslant l(\vb^{(k)})-l(\vb^{(k+1)})\rightarrow 0$ as $k\rightarrow\infty$. Thus,
$\|\vb^{(k+1)}-\vb^{(k)}\|\rightarrow0$. This theorem now follows immediately from Theorem 5 of Wu (1983).\end{proof}

Since the bridge, SCAD and MCP penalties all satisfy the condition that $S(l^*)$ is discrete, an OEM sequence for any of them converges to the stationary points of $l$.
Theorem \ref{th:conv} is obtained under the condition $\m{\Delta}' \m{\Delta}$ is not singular.
It is easy to show that Theorem \ref{th:conv} holds with probability one if the error $\ve$ in (\ref{lm}) has a continuous distribution.

We now derive the convergence rate of the OEM sequence in (\ref{mQ}). Following Dempster, Laird, and Rubin (1977), write
\begin{equation*}
\vb^{(k+1)}=\v{M}(\vb^{(k)}),
\end{equation*}
where the map $\v{M}(\vb)=(M_1(\vb),\ldots,M_p(\vb))'$ is defined by (\ref{mQ}). We capture the convergence rate of the OEM sequence $\{\vb^{(k)}\}$ through $\v{M}$.
Assume that (\ref{XcdI}) holds for $d\geqslant\gamma_1$, where $\gamma_1$
is the largest eigenvalue of $\m{X}'\m{X}$. 
For active orthogolization in Section \ref{sec:arc}, this assumption holds by taking $\m{S}=\m{I}_p$; see Remark \ref{rm:1}.

Let $\vb^*$ be the limit of the OEM sequence $\{\vb^{(k)}\}$. As in Meng (1994), we call
\begin{equation}\label{R}
R=\limsup_{k\rightarrow\infty}\frac{\|\vb^{(k+1)}-\vb^*\|}{\|\vb^{(k)}-\vb^*\|}
=\limsup_{k\rightarrow\infty}\frac{\|\v{M}(\vb^{(k)})-\v{M}(\vb^*)\|}{\|\vb^{(k)}-\vb^*\|},
\end{equation}the global rate of convergence for the OEM
sequence. If there is no penalty in (\ref{pls}), i.e., computing the OLS estimator, the global rate of convergence $R$ in  (\ref{R}) becomes the largest eigenvalue of
$\m{J}(\vb^*)$, denoted by $R_0$, where $\m{J}(\vp)$ is the $p\times p$ Jacobian matrix for $\v{M}(\vp)$ having $(i,j)$th entry $\partial M_i(\vp)/\partial \vp_j$. If
(\ref{XcdI}) holds, then $\m{J}(\vb^*)=\m{A}/d$ with $\m{A}=\m{\Delta}' \m{\Delta}$. Thus,
\begin{equation}
R_0=\frac{d-\gamma_p}{d}.\label{rate0}
\end{equation}

For (\ref{pls}), the penalty function $P(\vb;\lambda)$ typically is not sufficiently smooth and $R$ in (\ref{R}) has no analytic form. Theorem~\ref{th:rate} gives an
upper bound of $R_{\mathrm{net}}$, the value of $R$ for the elastic-net penalty in (\ref{netP}) with $\lambda_1,\lambda_2\geqslant0$.

\begin{theorem}\label{th:rate}
For $\m{\Delta}$ from (\ref{Xc}), if (\ref{XcdI}) holds, then $R_{\mathrm{NET}}\leqslant R_0$.
\end{theorem}

\begin{proof} Let $\v{x}_j$ denote the $j$th column of $n\times p$ matrix $\m{X}$ in \eqref{lm} and $\v{a}_j$ denote the $j$th column of $\m{A}=\m{\Delta}'\m{\Delta}$,
respectively.
For an OEM sequence for the elastic-net, by (\ref{net}),\begin{equation*}M_j(\vb)=f(\v{x}_j' \v{y}+\v{a}_j' \vb),\ \text{for}\ j=1,\ldots,p,\end{equation*} where
$$f(u)={\mathrm{sign}}(u)\left(\frac{|u|-\lambda_1}{d+\lambda_2}\right)_+.$$

For $j=1,\ldots,p$, observe that\begin{eqnarray*}\frac{|M_j(\vb^{(k)})-M_j(\vb^*)|}{\|\vb^{(k)}-\vb^*\|}&=&\frac{|f(\v{x}_j' \v{y}+\v{a}_j' \vb^{(k)}) -f(\v{x}_j'
\v{y}+\v{a}_j' \vb^*)|} {|(\v{x}_j' \v{y}+\v{a}_j' \vb^{(k)})-(\v{x}_j' \v{y}+\v{a}_j' \vb^*)|}\\&&\cdot\frac{|(\v{x}_j' \v{y}+\v{a}_j' \vb^{(k)})- (\v{x}_j'
\v{y}+\v{a}_j' \vb^*)|}{\|\vb^{(k)}-\vb^*\|}
\\&\leqslant&\frac{1}{d}\cdot\frac{|\v{a}_j' (\vb^{(k)}-\vb^*)|}{\|\vb^{(k)}-\vb^*\|}.\end{eqnarray*}
Thus, $$\frac{\|M(\vb^{(k)})-M(\vb^*)\|}{\|\vb^{(k)}-\vb^*\|}\leqslant\frac{1}{d}\cdot\frac{\|A(\vb^{(k)}-\vb^*)\|}{\|\vb^{(k)}-\vb^*\|}
\leqslant\frac{d-\gamma_p}{d}.$$This completes the proof.\end{proof}

\begin{remark}
Theorem~\ref{th:rate} indicates that, for the same $\m{X}$ and $\v{y}$  in (\ref{lm}), the OEM solution for the elastic-net numerically converges faster than its
counterpart for the OLS. Since the lasso is a special case of the elastic-net with $\lambda_2=0$ in (\ref{netP}), this theorem holds for the lasso as well.
\end{remark}

\begin{remark}From (\ref{rate0}) and Theorem \ref{th:rate}, the convergence rate of the OEM algorithm depends on the ratio of $\gamma_p$
and $d$ equal to or larger than $\gamma_1$. This rate is the fastest when $d=\gamma_1=\gamma_p$, i.e., if $\m{X}$ is orthogonal and standardized. This result suggests
that OEM converges faster if $\m{X}$ has controlled correlation like from a supersaturated design or a nearly orthogonal Latin hypercube design (Owen 1994).
\label{rm:r}
\end{remark}

\begin{figure}[t]\begin{center}\scalebox{0.6}[0.6]{\includegraphics{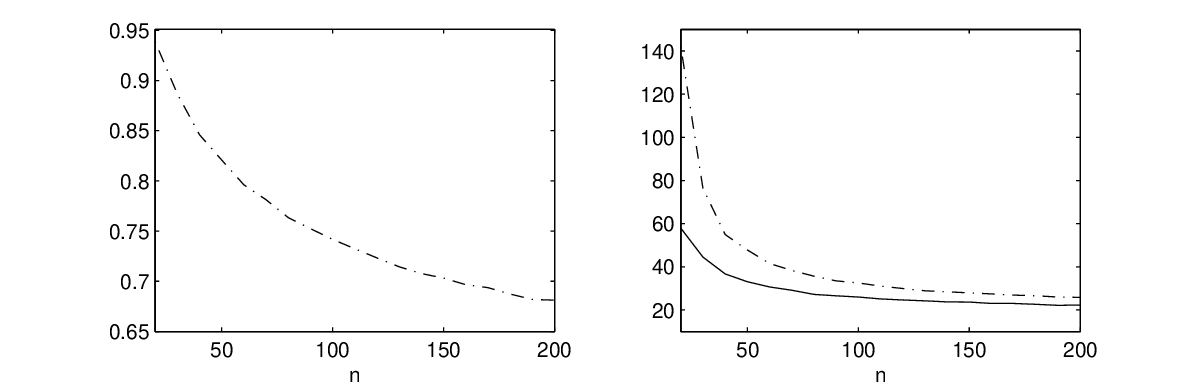}}\end{center}
\caption{(Left) the average values of $R_0$ in (\ref{rate0}) against increasing $n$ for Example~\ref{ex:rate}; (right) the average iteration numbers against increasing
$n$ for Example~\ref{ex:rate}, where the dashed and solid lines denote the OLS estimator and the lasso, respectively.} \label{fig:convR}
\end{figure}

\begin{example}\label{ex:rate}
We generate $\m{X}$ from $p$ dimensional Gaussian distribution $N(\v{0},\m{V})$ with $n$ independent observations, where the $(i,j)$th entry of $\m{V}$ is 1 for $i=j$
and $\rho$ for $i\neq j$. Values of $\v{y}$ and $\vb$ are generated by (\ref{lm}) and (\ref{beta}). The same setup was used in Friedman, Hastie, and Tibshirani (2009).
For $p=10,\ \rho=0.1,\ \lambda=0.5$ and increasing $n$, the left panel of Figure \ref{fig:convR} depicts the average values of $R_0$ in (\ref{rate0}) against increasing
$n$ and the right panel of the figure depicts the average iteration numbers against increasing $n$, with the dashed and solid lines corresponding to the OLS estimator
and the lasso, respectively. This figure indicates that OEM requires \emph{fewer} iterations as $n$ becomes larger, which makes OEM particulary attractive for situations
with big tall data. The OEM sequence for the lasso requires fewer iteration than its counterpart for the OLS, thus validating Theorem~\ref{th:rate}.
\end{example}

\section{POSSESSING GROUPING COHERENCE}
\label{sec:fulali} \hskip\parindent\vspace{-0.6cm}

Data with fully aliasing structures commonly appear in observational studies and designed experiments. Here we consider the convergence of the OEM algorithm
when the regression matrix $\m{X}$ in (\ref{lm}) is singular due to fully aliased columns. Let $\m{X}$ be standardized as in (\ref{stand}) with columns
$\v{x}_1,\ldots,\v{x}_p$. If $\v{x}_i$ and $\v{x}_j$ are fully aliased, i.e., $|\v{x}_i|=|\v{x}_j|$, then the objective function in (\ref{pls}) for the lasso is not
strictly convex and has many minima (Zou and Hastie 2005).

If some columns of $\m{X}$ are identical, it is desirable to have grouping coherence with the same regression coefficient. This is suggested by Zou and Hastie (2005) and
others. Definition~\ref{df1} makes this precise.

\begin{definition}\label{df1}
An estimator $\hat{\vb}=(\hat{\beta}_1,\ldots,\hat{\beta}_p)'$ of $\vb$ in (\ref{lm}) has \emph{grouping coherence} if $\v{x}_i=\v{x}_j$ implies
$\hat{\beta}_i=\hat{\beta}_j$ and $\v{x}_i=-\v{x}_j$ implies $\hat{\beta}_i=-\hat{\beta}_j$.
\end{definition}

Some penalties other than the lasso can produce estimators with grouping coherence (Zou and Hastie 2005; Bondell and Reich 2008; Tutz and Ulbricht 2009; Petry and
Tutz 2012), but they require more than one tuning parameters, which leads to more computational burden. \emph{Instead of changing the penalty, OEM can
give a lasso solution with this property}. This also holds for SCAD and MCP. Recall that $\hat{\vb}^*$ in \eqref{vb}, which can be obtained by OEM, has a stronger
property than grouping coherence.

Let $\v{0}_p$ denote the zero vector in ${\mathbb{R}}^p$. Let $\v{e}_{ij}^+$ be the vector obtained by replacing the $i$th and $j$th entries of $\v{0}_p$ with $1$. Let
$\v{e}_{ij}^-$ be the vector obtained by replacing the $i$th and $j$th entries of $\v{0}_p$ with $1$ and $-1$, respectively. Let ${\mathcal{E}}$ denote the set of all
$\v{e}_{ij}^+$ and $\v{e}_{ij}^-$. By Definition \ref{df1}, an estimator $\hat{\vb}$ has grouping coherence if and only if for any $\va\in{\mathcal{E}}$ with
$\m{X}\va=\v{0}$, $\va'\hat{\vb}=0$.

\begin{lemma}\label{lm2}
Suppose that (\ref{XcdI}) holds. For the OEM sequence $\{\vb^{(k)}\}$ of the lasso, SCAD or MCP, if $\m{X}\va=\v{0}$ and $\va'\vb^{(k)}=0$ for $\va\in{\mathcal{E}}$,
then $\va'\vb^{(k+1)}=0$. \end{lemma}

\begin{proof}
For $\v{u}$ in (\ref{vu}), $\va'\v{u}=\va'\m{X}'\v{y}+\va'(d\m{I}_p-\m{X}'\m{X})\vb^{(k)}=0$ for any $\va\in{\mathcal{E}}$ with $\m{X}\va=\v{0}$ and $\va'\vb^{(k)}=0$.
Then by (\ref{lasso}), (\ref{scad}) and (\ref{mcp}), an OEM sequence of the lasso, SCAD or MCP satisfies the condition that if $\va'\v{u}=0$, then $\va'\vb^{(k+1)}=0$
for $\va\in{\mathcal{E}}$. This completes the proof.\end{proof}

\begin{remark}Lemma \ref{lm2} implies that, for $k=1,2,\ldots$, $\vb^{(k)}$ has grouping coherence if $\vb^{(0)}$ has grouping coherence. Thus, if
$\{\vb^{(k)}\}$ converges, then its limit has grouping coherence. By Theorem \ref{th:conv}, if $d>\lambda_1$ in (\ref{XcdI}), then an OEM sequence for the SCAD or MCP
converges to a point with grouping coherence. \end{remark}

When $\m{X}$ in (\ref{lm}) has fully aliased columns, the objective function in (\ref{pls}) for the lasso has many minima and hence the condition in Theorem
\ref{th:convex} does not hold. Theorem \ref{th:gc} shows that, even with full aliasing, an OEM sequence (\ref{lasso}) for the lasso  converges to a point with grouping
coherence.

\begin{theorem} \label{th:gc}Suppose that (\ref{XcdI}) holds. If $\vb^{(0)}$ has grouping coherence, then as $k\rightarrow\infty$,
the OEM sequence $\{\vb^{(k)}\}$ of the lasso converges to a limit that has grouping coherence.
\end{theorem}

\begin{proof} Partition columns of $\m{X}$ in (\ref{lm}) as $(\m{X}_1\ \m{X}_2)$, where no two columns of $\m{X}_2$ are fully aliased and any
column of $\m{X}_1$ is fully aliased with at least one column of $\m{X}_2$. Let  $J$ denote the number of columns in $\m{X}_1$. Partition $\vb$ as $(\vb_1',\ \vb_2')'$
and $\vb^{(k)}$ as $(\vb_1^{(k)'},\ \vb_2^{(k)'})'$, corresponding to $\m{X}_1$ and $\m{X}_2$, respectively. For $j=1,\ldots,p$, let
$$\omega(j)=\#\{i=1,\ldots,p:|\v{x}_i|=|\v{x}_j|\}.$$ By Lemma \ref{lm2}, for $j=1,\ldots,J$, $\beta_j^{(k)}=\beta_{j'}^{(k)}$ if $\v{x}_j=\v{x}_{j'}$ and
$\beta_j^{(k)}=-\beta_{j'}^{(k)}$ otherwise, where $j'\in\{J+1,\ldots,p\}$. It follows that $\{\vb_2^{(k)}\}$ is an OEM sequence for solving
\begin{equation}\min_{{\scriptsize\vt}}\|\v{y}-\tilde{\m{X}}\vt\|^2+2\sum_{j=1}^{p-J}|\theta_j|,\label{sof}\end{equation}
where $\vt=(\theta_1,\ldots,\theta_{p-J})'$, and the columns of $\tilde{\m{X}}$ are $\omega(J+1)\v{x}_{J+1},\ldots,\omega(p)\v{x}_{p}$. Because the objective function in
(\ref{sof}) is strictly convex, by Theorem \ref{th:convex}, $\{\vb_2^{(k)}\}$ converges to a limit with grouping coherence. This completes the proof.\end{proof}

\section{ACHIEVING THE ORACLE PROPERTY WITH NONCONVEX PENALTIES}\label{sec:op}
\hskip\parindent\vspace{-0.6cm}

Fan and Li (2001) introduced an important concept called the oracle property and showed that there exists one local minimum of the SCAD problem with this property when
$p$ is fixed. The corresponding results with a diverging $p$ were presented in Fan and Peng (2004) and Fan and Lv (2011). Because the optimization problem in (\ref{pls})
with the SCAD penalty has an exponential number of local optima (Huo and Ni 2007; Huo and Chen 2010), no theoretical results in the current literature, as far as we are aware, show that an
existing algorithm can provide such a local minimum. Zou and Li (2008) proposed the local linear approximation (LLA) algorithm to solve the SCAD problem and showed that
the one-step LLA estimator has the oracle property with a good initial estimator for a fixed $p$. The LLA estimator is not guaranteed to be a local minimum of SCAD. In contrast, we prove that the OEM solution to the
SCAD or MCP can achieve this property. Like Fan and Peng (2004) and Fan and Lv (2011), we allow $p$ to depend on $n$, which covers the fixed $p$ case as a special case.

Suppose that the number of nonzero coefficients of $\vb$ in (\ref{lm}) is $p_1$ (with $p_1\leqslant p$) and partition $\vb$ as
\begin{equation}\label{pb}
\vb=(\vb_1',\vb_2')',
\end{equation}
where $\vb_2=\v{0}$ and no component of $\vb_1$ is zero. Divide columns of the regression matrix $\m{X}$ in (\ref{lm}) to $(\m{X}_1\ \m{X}_2)$ with $\m{X}_1$
corresponding to $\vb_1$. A regularized least squares estimator of $\vb$ in (\ref{lm}) has the oracle property if it can not only select the correct submodel
asymptotically, but also estimate the nonzero coefficients $\vb_1$ in (\ref{pb}) as efficiently as if the correct submodel were known in advance. Specifically, an
estimator $\hat{\vb}=(\hat{\vb}_1',\hat{\vb}_2')'$ has this property if $\P(\hat{\vb}_2=\v{0})\to1$ and $\hat{\vb}_1-\vb_1$ follows a normal distribution $N(\v{0},\
\sigma^2(\m{X}_1'\m{X}_1)^{-1})$ asymptotically.

We now consider the oracle property of OEM sequences for SCAD. First we prove that, under certain conditions, a fixed point of the OEM iterations for SCAD can possess the oracle property. Here in after, $p$ depends on $n$ but $p_1$ and $\vb_1$ are
fixed for simplicity. A definition and several assumptions are needed.

\begin{definition} For a series of numbers $c_n\to\infty$ and a positive constant $\kappa$, an estimator $\hat{\vb}$ of $\vb$ is
said to be $c_n$-concentratively consistent of order $\kappa$ to $\vb$ if as $n\to\infty$, (i) $\|\hat{\vb}-\vb\|=O_p(c_n^{-1})$; (ii)
$\P(c_{n}\|\hat{\vb}-\vb\|\geqslant h_n)=O(\exp(-\delta h_n^{\kappa}))$ for any $h_n\to+\infty$, where $\delta>0$ is a constant.\label{def:ini}\end{definition}

\begin{assumption} The random error $\ve$ follows a normal distribution $N(\v{0},\ \sigma^2\m{I}_n)$.
\label{as:norm}\end{assumption}

\begin{assumption} The matrix $\m{X}/\sqrt{n}$ is standardized such that each entry on the diagonal
of $\m{X}'\m{X}/n$ is 1, and $\m{X}'\m{X}/n+\m{\Delta}'\m{\Delta}=d_n\m{I}_p$ with $d_n\geqslant\gamma_1$, where $\gamma_1$ is the largest eigenvalue of $\m{X}'\m{X}/n$.
\label{as:stand}\end{assumption}

In active orthogonalization, $d_n$ in Assumption \ref{as:stand} can take any number equal to or larger than $\gamma_1$.

\begin{assumption} As $n\rightarrow\infty$, $$\frac{\m{X}_1'\m{X}_1}{n}\rightarrow\m{\Sigma}_1,$$
where $\m{\Sigma}_1$ is a $p_1\times p_1$ positive definite matrix.\label{as:Z}\end{assumption}

\begin{assumption} The tuning parameter $\lambda=\lambda_n$ in (\ref{scadP}), $d_n$ in Assumption \ref{as:stand}, and $p$ satisfy the condition that,
as $n\rightarrow\infty$, $\lambda_n/n\rightarrow0$ and
$p\exp\big(-v(c_n\lambda_n/(nd_n))^\kappa\big)\rightarrow\infty$ for any $v>0$.\label{as:lam}\end{assumption}

For a fixed $p$, the OLS estimator is concentratively consistent with $c_n=\sqrt{n}$ and $\kappa=1$ in Definition \ref{def:ini} under Assumption \ref{as:norm}.
Generally, $c_n$ in the above assumptions satisfies $c_n=O(\sqrt{n})$. For example, $c_n=\sqrt{n/\log(p)}$ in the consistency analysis for the lasso (B\"{u}hlmann and
van de Geer 2011). Let $d_n=O(n^{q_1})$. To satisfy Assumption \ref{as:lam}, $q_1$ must be smaller than $1/2$. Note that $d_n\geqslant\gamma_1\geqslant p/n$. Therefore
if we set $p=n^{q}$ for some $q\geqslant0$, $q$ must be smaller than $3/2$. In other words, our results in this section can handle dimensionality of order $p=O(n^q)$ for
$q\in[0,3/2)$. For such a $q_1$, we can take the tuning parameter $\lambda_n\sim n^{q_2}$ to satisfy Assumption \ref{as:lam}, where $q_2\in(q_1+1/2,\ 1)$.

\begin{theorem} Let $\hat{\vb}^{\mathrm{f}}$ be a fixed point of the OEM iterations for SCAD with a fixed $a>2$ in (\ref{scadP}). Suppose that $\hat{\vb}^{\mathrm{f}}$
is $c_n$-concentratively consistent of order $\kappa$ to $\vb$ with $c_n=O(\sqrt{n}d_n)$ and $\kappa\leqslant2$. Under Assumptions \ref{as:norm}-\ref{as:lam}, as
$n\rightarrow\infty$, \\(i) $\P(\hat{\vb}_2^{\mathrm{f}}=\v{0})\rightarrow1$;
\\(ii) $\sqrt{n}(\hat{\vb}_1^{\mathrm{f}}-\vb_1)\rightarrow N(\v{0},\sigma^2\m{\Sigma_1^{-1}})$ in distribution.
\label{th:orlarp}\end{theorem}

The proof of Theorem \ref{th:orlarp} is deferred to the Appendix. This theorem indicates that a fixed point of OEM consistent to the true parameter is an oracle
estimator asymptotically even when $p$ grows faster than $n$. If we do not know whether a fixed point is consistent, with an initial point concentratively consistent to $\vb$, an OEM sequence can converge to that fixed point and possess the oracle property.

Let $\{\vb^{(k)},\ k=0,1,\ldots,\}$ be the OEM sequence from (\ref{scad}) for the SCAD with a fixed $a>2$ in (\ref{scadP}). Let $\eta_n$ be the largest eigenvalue of
$\m{I}_{p_1}-\m{X}_1'\m{X}_1/(nd_n)$. Clearly, $\eta_n\in (0,1)$. We need an assumption on $k=k_n$.

\begin{assumption}As $n\rightarrow\infty$, $d_n\eta_n^{k}\rightarrow0$, $k^3\exp(-v_1c_n^{\kappa})\rightarrow0$,
and $pk^3\exp\big(-v_2(c_n\lambda_n/({n}d_n))^\kappa\big)\rightarrow0$ for any $v_1,v_2>0$. \label{as:k}\end{assumption}

As $n\rightarrow\infty$, $d_n\eta_n^{k}\rightarrow0$ implies $k\to\infty$. In fact, $k$ can grow much faster than $n$. For example, suppose that $c_n=\sqrt{n/\log(n)}$
and $d_n=O(n^{q_1})$, where $q_1\in[0,1/2)$. Take $\lambda_n\sim n^{q_2}$, where $q_2\in(q_1+1/2, 1)$. With $p=O(n^{q})$ for any $q\in[0,3/2)$, one choice for $k$ to
satisfy Assumption \ref{as:k} is $k=\exp(n^{q_3})\ \text{for some}\ q_3\in\big(0,\ \kappa(q_2-q_1-1/2)\big).$

Under the above assumptions, Theorem~\ref{th:orp} shows that $\vb^{(k)}=(\vb^{(k)'}_1,\vb^{(k)'}_2)'$ can achieve the oracle property.

\begin{theorem}If $\vb^{(0)}$ is $c_n$-concentratively consistent of order $\kappa$ to $\vb$ with $c_n=O(\sqrt{n}d_n)$ and $\kappa\leqslant2$. Under Assumptions \ref{as:norm}-\ref{as:k}, as $n\rightarrow\infty$, \\(i)
$\P(\vb_2^{(k)}=\v{0})\rightarrow1$;
\\(ii) $\sqrt{n}(\vb_1^{(k)}-\vb_1)\rightarrow N(\v{0},\sigma^2\m{\Sigma_1^{-1}})$ in distribution.
\label{th:orp}\end{theorem} The proof of Theorem \ref{th:orp} is deferred to the Appendix.

\begin{remark}From (\ref{Ao}) in the proof of Theorem \ref{th:orp}, for any $k=1,2,\ldots$, $\vb^{(k)}$ is \emph{consistent} in variable selection. That is,
$\P(\beta^{(k)}_j\neq0\ \text{for}\ j=1,\ldots,p_1)\rightarrow1$
and $\P(\vb_2^{(k)}=\v{0})\rightarrow1$ as $n\rightarrow\infty$.\end{remark}

\begin{remark}
The proof of Theorem \ref{th:orp} uses the convergence rates of $\P(A_k)$ and $P(B_k)$. If an OEM sequence satisfies the condition that $\beta_j^{(k+1)}=0$ when
$|u_j|<\lambda$ and $\beta_j^{(k+1)}=u_j/d$ when $|u_j|>c\lambda$ for some positive constant $c$, then $\P(A_{k+1})=\P(|u_j|<\lambda)$ and
$\P(B_{k+1})=\P(|u_j|>c\lambda)$. Since an OEM sequence for MCP satisfies the above condition, an argument similar to the proof in the Appendix shows that the
convergence rates of $\P(A_k)$ and $\P(B_k)$ for MCP are the same as those with the SCAD. Thus, under Assumption \ref{as:norm}-\ref{as:k}, Theorem~\ref{th:orp} holds for
MCP with a fixed $a>1$ in (\ref{mcpP}).
\end{remark}

\begin{remark}
With minor modifications, Theorem \ref{th:orlarp} and \ref{th:orp} can allow $p_1$ to tend to infinity at a relatively low rate. They also hold if the normality
condition $\ve\sim N(\v{0}_n,\ \sigma^2\m{I}_n)$ is replaced by weaken conditions such as the sub-Gaussian condition (see e.g. Zhang 2010).
\end{remark}

Theorem \ref{th:orlarp} and \ref{th:orp} can handle dimensionality of order $p=O(n^q)$ for $q<3/2$. For $p$ exceeding this order, penalized regression methods
can perform poorly. A practical approach is a two-stage procedure like that in Fan and Lv (2008). The first stage uses an efficient
screening method to reduce the dimensionality. OEM can be used in the second stage to obtain a SCAD estimator with the oracle property.

The initial point in OEM for nonconvex penalties can be chosen as the OLS estimator if $p<n$. Otherwise, the lasso estimator, which is consistent under certain
conditions (Meinshausen and Yu 2009; B\"{u}hlmann and van de Geer 2011), can be used as the initial point.

Huo and Chen (2010) showed that, for the SCAD penalty, solving the global minimum of the SCAD problem leads to an NP-hard problem. Theorem \ref{th:orp} indicates that as
far as the oracle property is concerned, the local solution given by OEM will suffice.

\section{NUMERICAL ILLUSTRATIONS FOR SOLVING PENALIZED LEAST SQUARES}\label{sec:nc}

Existing algorithms for solving the regularized least squares problem in (\ref{pls}) include those in Fu (1998), Grandvalet (1998), Osborne, Presnell, and Turlach (2000), the LARS algorithm in Efron, Hastie, Johnstone, and
Tibshirani (2004) and the coordinate descent (CD) algorithm (Tseng 2001; Friedman, Hastie, Hofling and Tibshirani 2007; Wu and Lange 2008; Tseng and Yun 2009). The corresponding \texttt{R} packages include \texttt{lars} (Hastie and Efron 2011), \texttt{glmnet} (Friedman, Hastie, and Tibshirani 2011), and \texttt{scout} (Witten and
Tibshirani 2011). For nonconvex penalties like SCAD and MCP, existing algorithms include local quadratic approximation (Fan and
Li 2001; Hunter and Li 2005), local linear approximation (Zou and Li 2008), the CD algorithm (Breheny and Huang 2011; Mazumder, Friedman, and Hastie
2011) and the minimization by iterative soft thresholding algorithm (Schifano, Strawderman, and Wells 2010), among others. Different from these algorithms, OEM
handles each dimension of the iterated vector separably and equally as in (\ref{iols}), and has appealing features such as
grouping coherence in Section \ref{sec:fulali} and the oracle property in Section \ref{sec:op}. Putting these properties aside, one may be interested in numerical comparisons of OEM and other algorithms. Here we compare OEM with the CD and LARS algorithms for regularized least squares.

\subsection{COMPARISONS WITH OTHER ALGORITHMS}\label{subsec:ccd}

\subsubsection{GROUPING COHERENCE}\label{subsubsec:group}

We illustrate grouping coherence of OEM in Section \ref{sec:fulali} with a simulated data set of four predictors, where the variables $X_1$ and
$X_2$ are generated from independent standard normal distributions. The degenerated design matrix is formulated by $X_3 = -X_1$ and $X_4 = -X_2$, where the predictors
consist of two pairs of perfectly negative correlated random variables. The true relationship between the response and predictors is
\begin{eqnarray*}
 y = -X_3 - 2X_4.
\end{eqnarray*}

\begin{figure}[thb]
\begin{center}
 \includegraphics[scale = 0.6]{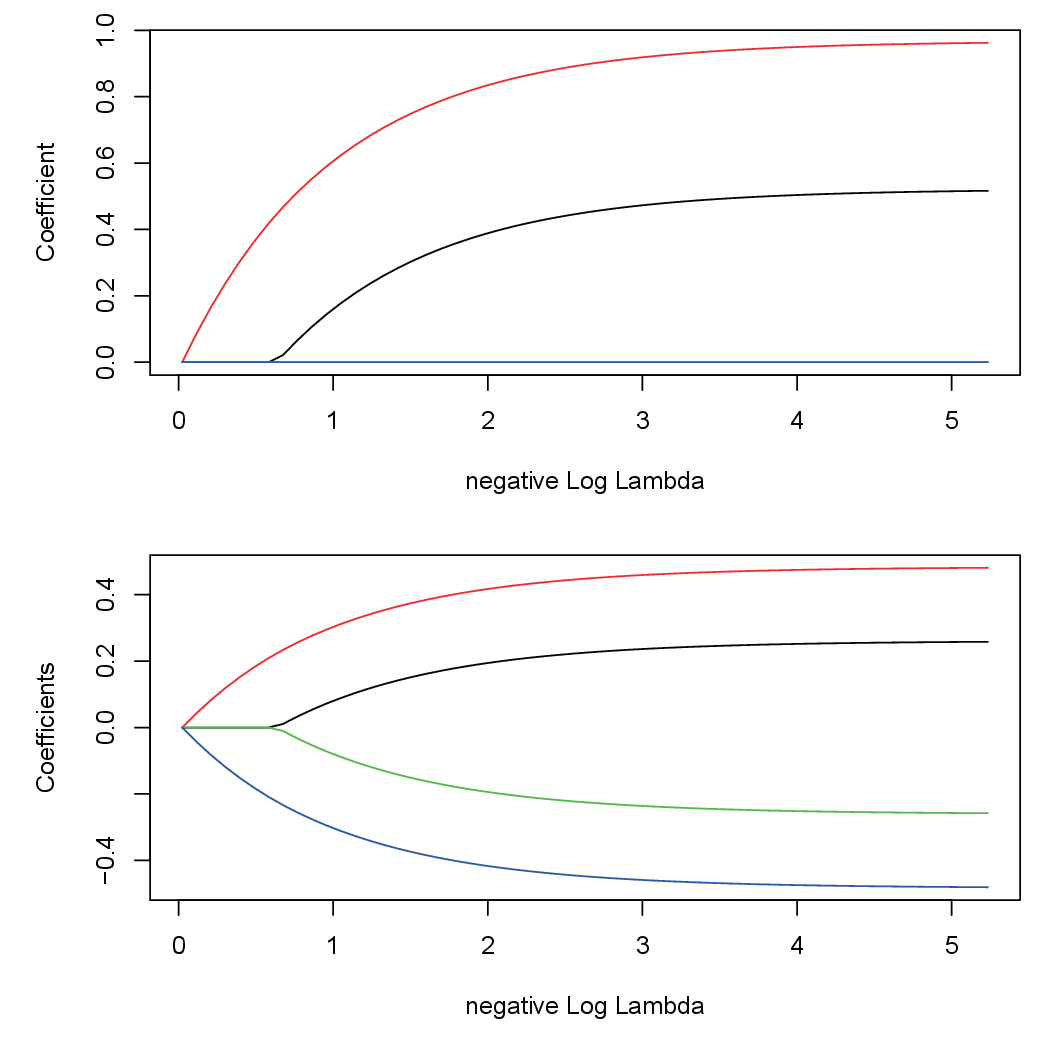}
\caption{Solution paths of the lasso fitted by CD (the upper panel) and OEM (the lower panel) in Section \ref{subsubsec:group}.} \label{path1}
\end{center}
\end{figure}

Figure \ref{path1} displays the solution paths for the data using the lasso fitted by \texttt{R} packages \texttt{glmnet} and \texttt{oem} on the same set of tuning
parameters $\lambda$. The package \texttt{lars} gives the same solution path as \texttt{glmnet}. This figure reveals that OEM estimates the perfectly negative correlated
pairs to have exactly the opposite signs but CD only has $X_1$ and $X_2$ in the model and fixes $X_3$ and $X_4$ to be zero for any $\lambda$. This difference is due to
the fact that in every iteration, both CD and LARS will find the predictor with the largest improvement on the target function and if more than one coordinates can give
better results, only the one with the smallest index will enter the model. In contrast, OEM considers all the predictors in every iteration equally, so the ones with
same contribution to the target will receive equal steps.

\begin{figure}[thb]
\begin{center}
 \includegraphics[scale = 0.6]{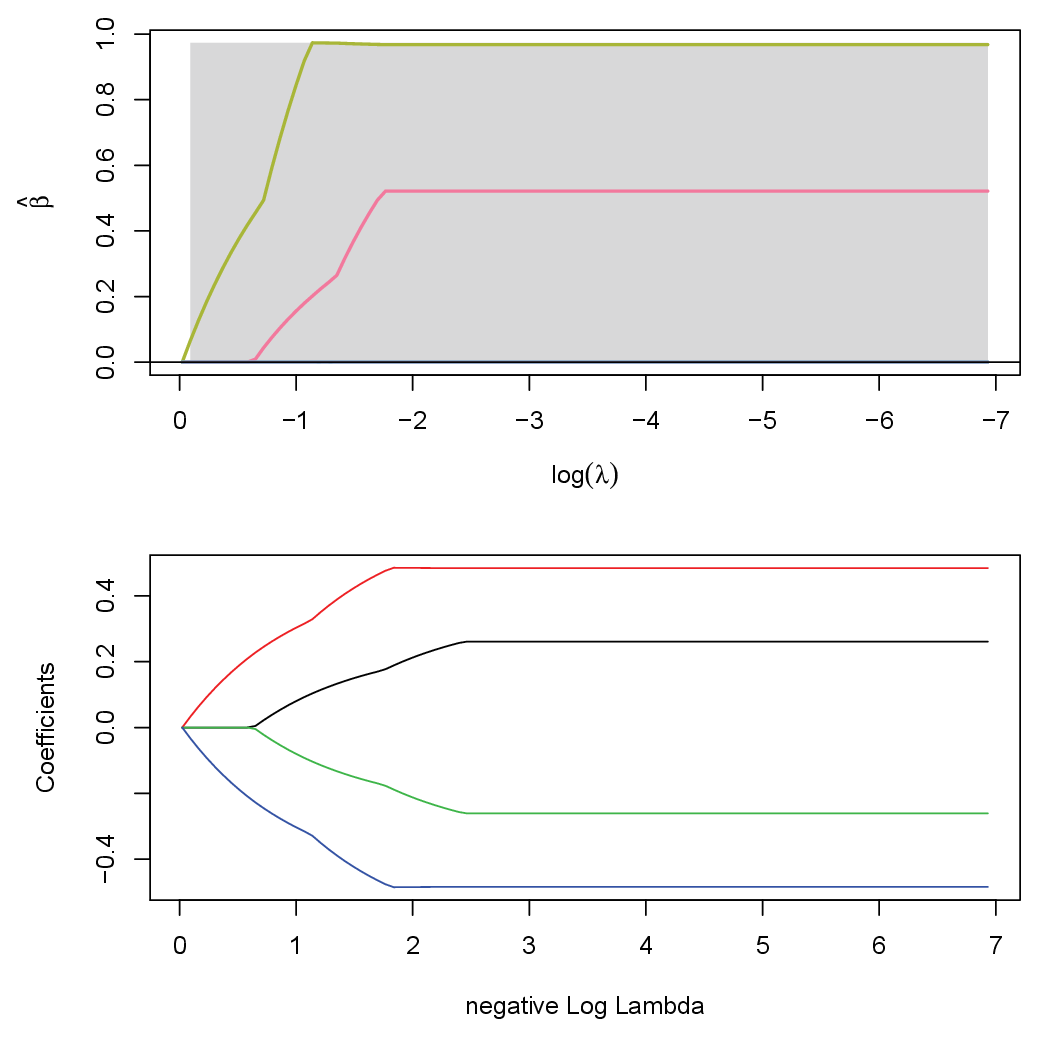}
\caption{Solution paths of SCAD fitted by CD (from package \texttt{ncvreg}) in the upper panel and OEM for the lower panel in Section \ref{subsubsec:group}..} \label{path2}
\end{center}
\end{figure}

Grouping coherence of OEM also holds for non-convex penalties such as SCAD, with the solution paths shown in Figure \ref{path2}, where the same data used above for the lasso is used.

\subsubsection{SPEED}\label{subsubsec:speed}

We now compare the computational efficiency of OEM for regularized least squares problems with the coordinate descent (CD) algorithm, which is considered the fastest
among the current choices. OEM is implemented in \texttt{R} package \texttt{oem} with main code in \texttt{C++}. For fitting the lasso, we compare OEM and the \texttt{R}
package \texttt{glmnet}, which has the main code in \texttt{Fortran} and uses several tricks to speed up. We found that \texttt{glmnet} is faster than \texttt{oem} in most scenarios, but \texttt{oem} has grouping
coherence; see Section \ref{subsubsec:group}. Next, we focus on comparisons of OEM and the package \texttt{ncvreg} developed in \texttt{C} for SCAD and MCP
penalties.

We first consider the situation when the sample size $n$ is larger than the number of variables $p$. Three different covariance matrix structures for the
predictor variables are compared. The first is the case where all the variables are independently generated from standard normal distribution, the second and  third
cases involve design matrices with a correlation structure
\begin{eqnarray}
 \text{Cor}(X_i, X_j) = \rho^{|i - j|}\;\; \text{for}\;i,j = 1, \ldots, p,\label{Xcorr}
\end{eqnarray}
where $\rho = 0.2, 0.8$. The response is generated independent of the design matrix and the true model is $\v{y} = \ve$, where $\ve$ follows the normal
distribution $N(\v{0}, \sigma^2 \m{I}_n)$.

\begin{table}[htb]
\begin{center}
\begin{tabular}{cc||ccc|ccc}
\multirow{2}{*}{$p$} & \multirow{2}{*}{$n$} &\multicolumn{3}{|c}{OEM}&\multicolumn{3}{|c}{CD}\\
\cline{3-8}
&&$\rho = 0$&$\rho = 0.2$&$\rho = 0.8$&$\rho = 0$&$\rho = 0.2$&$\rho = 0.8$\\
\hline
\hline\multirow{3}{*}{ $20$ }& $400$ & $0.0052$ & $0.0059$ & $0.0240$ & $0.0451$ & $0.0245$ & $0.0209$ \\
& $1000$ & $0.0061$ & $0.0073$ & $0.0262$ & $0.0449$ & $0.0516$ & $0.0452$ \\
& $2000$ & $0.0088$ & $0.0099$ & $0.0277$ & $0.0826$ & $0.0927$ & $0.0844$ \\
\hline\multirow{3}{*}{ $50$ }& $1000$ & $0.0189$ & $0.0261$ & $0.1803$ & $0.1398$ & $0.1437$ & $0.1797$ \\
& $2500$ & $0.0311$ & $0.0380$ & $0.1918$ & $0.4483$ & $0.4808$ & $0.4613$ \\
& $5000$ & $0.0609$ & $0.0689$ & $0.2291$ & $0.833$ & $0.9233$ & $0.8912$ \\
\hline\multirow{3}{*}{ $100$ }& $2000$ & $0.0946$ & $0.1193$ & $1.0037$ & $0.8865$ & $0.8612$ & $0.9964$ \\
& $5000$ & $0.1689$ & $0.2085$ & $1.1002$ & $2.2004$ & $2.4043$ & $2.691$ \\
& $10000$ & $0.4551$ & $0.5342$ & $1.2832$ & $4.8513$ & $5.6488$ & $7.7149$
\end{tabular}
\caption{Average runtime (second) comparison between OEM and CD for SCAD when $n$ is larger than $p$} \label{comp1}
\end{center}
\end{table}
To compare the performance of the OEM and CD algorithms for SCAD penalty, data are generated 10 times and the average runtime are given in Table \ref{comp1}.
The table indicates that OEM has advantages when the sample size is significantly larger than the number of variables especially for the independent design. Both
algorithms require more fitting time when the correlations among the covariates increase.
Table \ref{comp2} compares the algorithms with large $p$ small $n$. It turns out that the CD algorithm is faster and the
computational gap gets wider when the ratio of $p/n$ increases. Since regularized least squares methods are usually more efficient after the dimensionality $p$ is reduced from very large to moderate by a
screening procedure (Fan and Lv 2008), a remedy for this drawback is to use OEM after screening the important variables.

\begin{table}[htb]
\begin{center}
\begin{tabular}{cc||ccc|ccc}
\multirow{2}{*}{$p$} & \multirow{2}{*}{$n$} &\multicolumn{3}{|c}{OEM}&\multicolumn{3}{|c}{CD}\\
\cline{3-8}
&&$\rho = 0$&$\rho = 0.2$&$\rho = 0.8$&$\rho = 0$&$\rho = 0.2$&$\rho = 0.8$\\
\hline
\hline\multirow{3}{*}{ $200$ }& $100$ & $0.8425$ & $0.9703$ & $1.2412$ & $0.1430$ & $0.1584$ & $0.1505$ \\
& $200$ & $6.2634$ & $6.784$ & $8.4708$ & $0.4800$ & $0.4728$ & $0.462$ \\
& $400$ & $3.0315$ & $3.2366$ & $6.7653$ & $0.8429$ & $0.8311$ & $1.0044$ \\
\hline\multirow{3}{*}{ $500$ }& $250$ & $4.7629$ & $5.1432$ & $7.0855$ & $1.3622$ & $1.421$ & $1.1643$ \\
& $500$ & $51.338$ & $51.941$ & $57.536$ & $4.4924$ & $4.4082$ & $3.4217$ \\
& $1000$ & $31.070$ & $32.631$ & $54.069$ & $10.175$ & $9.0097$ & $9.8956$ \\
\hline\multirow{3}{*}{ $1000$ }& $500$ & $7.8277$ & $8.6695$ & $23.833$ & $8.5252$ & $8.1363$ & $7.2247$ \\
& $1000$ & $741.54$ & $978.13$ & $1063.7$ & $64.216$ & $67.935$ & $45.511$ \\
& $2000$ & $658.19$ & $676.82$ & $739.18$ & $152.80$ & $129.01$ & $100.25$ \\
\hline\multirow{3}{*}{ $1200$ }
& $100$ & $14.313$ & $12.049$ & $14.722$ & $0.9061$ & $0.8197$ & $0.9102$ \\
& $150$ & $20.443$ & $15.972$ & $18.676$ & $1.8636$ & $1.3811$ & $1.3246$ \\
& $240$ & $24.885$ & $20.313$ & $24.714$ & $3.6128$ & $2.7939$ & $2.5308$
\end{tabular}
\caption{Average runtime (second) comparison between OEM and CD for SCAD for large $p$} \label{comp2}
\end{center}
\end{table}
\subsection{PERFORMANCE COMPARISONS WITH ONE-STEP ESTIMATOR}\label{subsec:pc}

We compare the SCAD solution computed by OEM with Zou and Li (2008)'s one-step LLA estimator. The model used here is
\begin{equation}Y=\sum_{j=1}^p\beta_jX_j+\varepsilon, \label{mfs}\end{equation}where $X_j$'s are generated from (\ref{Xcorr}), $\varepsilon\sim N(0,\sigma^2)$, $p=8$,
and \begin{equation}\vb=(\beta_1,\ldots,\beta_8)'=(3, 1.5, 0, 0, 2, 0, 0, 0)'.\end{equation}The sample size is fixed as 60. We first use the OEM algorithm to compute the
SCAD solution with the initial point being the OLS estimator. The tuning parameter $\lambda$ in (\ref{scadP}) is selected by BIC (Wang, Li, and Tsai 2007). With the same
$\lambda$, we compute the one-step estimator, and compare the variable selection errors (VSEs) and the model errors (MEs) of the two estimators. The VSE and ME of an estimator $\hat{\vb}$ are respectively defined as $$\mathrm{VSE}(\hat{\vb})=|\{j:\ j\in\mathcal{A}(\vb)\ \text{but}\ j\notin\mathcal{A}(\hat{\vb})\}|+|\{j:\
j\in\mathcal{A}(\hat{\vb})\ \text{but}\ j\notin\mathcal{A}(\vb)\}|$$and$$\mathrm{ME}(\hat{\vb})=(\hat{\vb}-\vb)'(\m{X}'\m{X})(\hat{\vb}-\vb)/n,$$where $|\cdot|$ denotes
cardinality and $\mathcal{A}(\vb)=\{j:\ \beta_j\neq0,\ j=1,\ldots,p\}$.

\begin{table}[htb]
\begin{center}
\begin{tabular}{lc||ccc|ccc}
 & \multirow{2}{*}{$\sigma$} &\multicolumn{3}{|c}{OEM}&\multicolumn{3}{|c}{one-step}\\
\cline{3-8}
&&$\rho = 0$&$\rho = 0.5$&$\rho = 0.9$&$\rho = 0$&$\rho = 0.5$&$\rho = 0.9$\\
\hline \hline\multirow{2}{*}{ VSE } & $1$ & 1.487 (1.67) & 1.111 (1.36) & 1.420 (0.67)
      & 1.730 (1.92) & 1.441 (1.73) & 3.550 (1.14) \\
& $3$ & 3.448 (1.12) & 3.060 (1.09) & 3.614 (1.26)
      & 3.408 (1.18) & 3.294 (1.17) & 4.474 (0.96) \\
\hline\multirow{2}{*}{ ME } & $1$ & 0.091 (0.06) & 0.084 (0.05) & 0.076 (0.07)
      & 0.136 (0.09) & 0.123 (0.09) & 0.138 (0.14) \\
& $3$ & 1.043 (0.56) & 1.048 (0.61) & 1.168 (0.63)
      & 1.070 (0.56) & 1.090 (0.60) & 1.207 (0.64) \\
\end{tabular}
\caption{Average VSEs and MEs of OEM and the one-step estimator (standard deviations in parentheses)} \label{c_OEM_ONE}
\end{center}
\end{table}

The average VSE and ME values of the two estimators over 1000 times are given in Table \ref{c_OEM_ONE}. The SCAD estimator computed by OEM outperforms the one-step
estimator in most cases, especially when $\rho$ is large.

\subsection{REAL DATA EXMAPLE}\label{subsec:real}

Consider a dataset from US Census Bureau County and City Data Book 2007. The response is population change in percentage. The covariates include
\begin{enumerate}
 \item Economic variables like income per capita, household income, poverty.
 \item Population distribution like percentages of different races, education levels.
 \item Crime rates like violent crimes and property thefts.
 \item Miscellaneous variables like Republic, Democratic, death and birth rates.
\end{enumerate}
These variables are in percentage of population of the individual counties.

There are $2573$ (counties) observations without missing observations. The linear regression model in (\ref{lm}) is used to fit the data. The solution paths for the
lasso, SCAD and MCP fitted to the data set are given in Figure \ref{fig:path3}. The number of non-zero coefficients, cross validation residual sum of squares, AIC and
BIC are presented in Table \ref{comp5}, where the tuning parameter $\lambda$ is chose by BIC. The selected significant variables include
\begin{itemize}
 \item Percentage of Household income above $750,000$ dollars, which has large positive effect on the percentage of population change.
 \item Social security program beneficiaries. The larger the number of beneficiaries in the program, the higher the population change.
 \item Both the percentages of retired people and under 18 years old have negative effects since they are major sources of migrants leaving the county.
 \item Birth and death rate with positive and negative effects, respectively.
\end{itemize}

The significant variables reveal that the population change is highly related to the living standards of the counties. Table \ref{comp5} compares the fitted models from different regularized least squares
problems. Note that MCP has the most sparse model with little sacrifice of CV error, AIC and BIC scores. LASSO has the model with smallest CV error but including nearly
all the candidate predictors. In the example, the regularized models favor complex models with many nonzero coefficients and this reveals the fact that there are
many factors that have profound influence on population change of counties in the US. In addition, the last two columns of Table \ref{comp5} also give the runtime of
fitting the 10-fold cross-validation to the data, where OEM is implemented in the R package \texttt{oem}, LASSO with CD from \texttt{glmnet}, and SCAD and MCP from
\texttt{ncvreg}.

\begin{figure}[phb]
\begin{center}\scalebox{0.8}[0.4]{\includegraphics{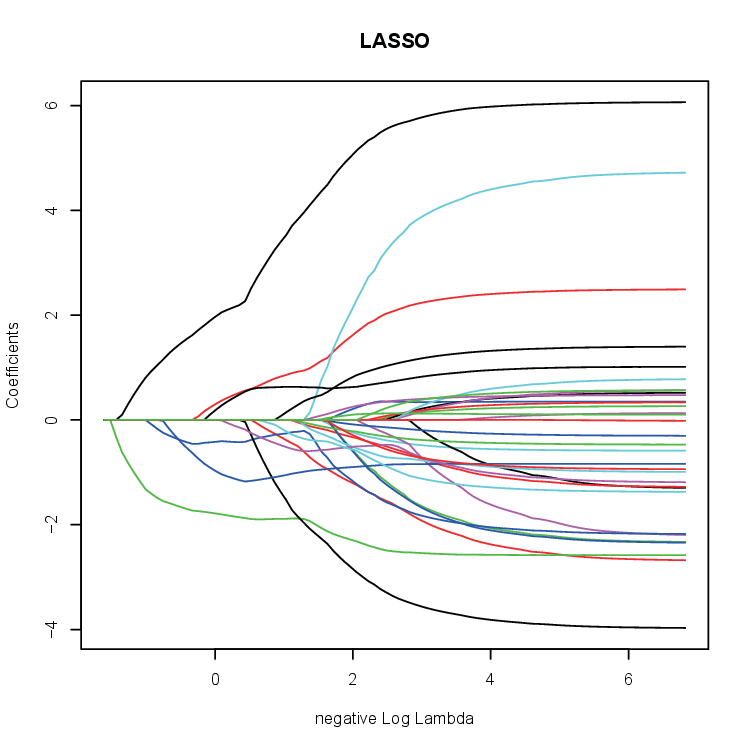}}\\
\scalebox{0.8}[0.4]{\includegraphics{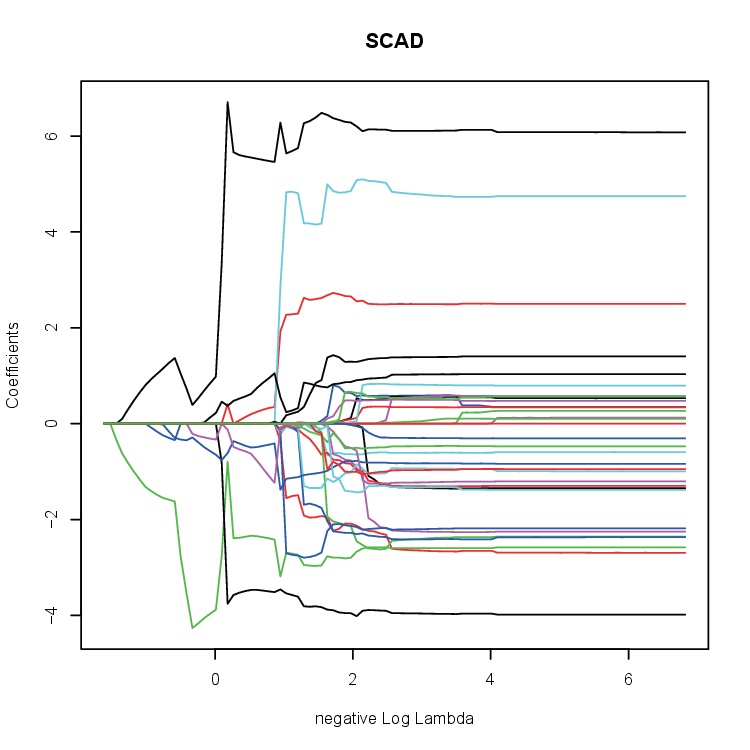}}\\
\scalebox{0.8}[0.4]{\includegraphics{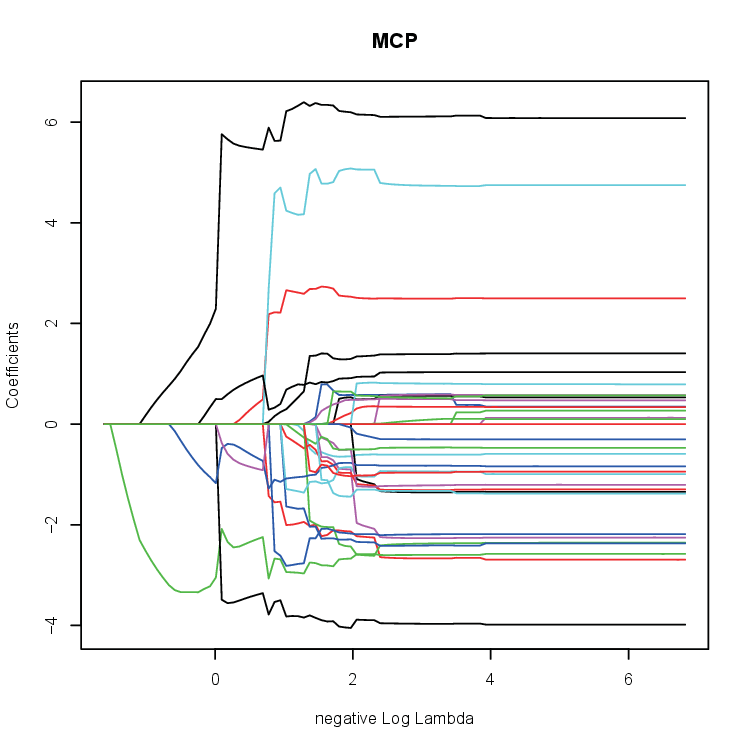}}\end{center} \caption{Solution paths for LASSO, SCAD and MCP for US census bureau data.}\label{fig:path3}
\end{figure}

\begin{table}[htb]
\begin{center}
\begin{tabular}{c||cccc|cc}
\multirow{2}{*}{Penalty}  & \multicolumn{4}{|c}{Final Model} & \multicolumn{2}{|c}{Runtime (s)}\\
& Size & CV error & AIC & BIC & OEM & CD \\
\hline
LASSO & $32$ & $46.93$ & $3.81$ & $3.87$ & $2.097$ & $0.273$\\
SCAD & $28$ & $47.12$ & $3.81$ & $3.87$ & $1.783$ & $3.454$\\
MCP & $23$ & $47.17$ & $3.82$ & $3.88$ & $1.433$ & $3.032$
\end{tabular}
\caption{Lasso, SCAD and MCP results for the U.S. Census Bureau data} \label{comp5}
\end{center}
\end{table}

\section{DISCUSSION}\label{sec:diss}
\hskip\parindent\vspace{-0.6cm}

We have proposed a new algorithm called OEM for solving ordinary and regularized least squares problems with general data structures. OEM has unique theoretical properties, including convergence to the
Moore-Penrose generalized inverse-based least squares estimator for singular regression matrices and convergence to a point having grouping coherence for the lasso, SCAD or MCP. Different from existing algorithms, OEM can provide a local solution with the oracle property for the SCAD and MCP penalties. This suggests a new interface between optimization and statistics for regularized methods.

OEM is very fast for big tall data with $n>p$, such as the data deluge in astronomy, the Internet and marketing (the Economist 2010), large-scale industrial experiments (Xu 2009) and modern simulations in engineering (NAE 2008). For applications for big wide data with $n<p$ like micro-array, OEM is generally slow. We can use a two-stage procedure like that in Fan and Lv (2008) to mitigate this drawback. The first stage uses an efficient screening method to reduce the dimensionality. OEM can be used in the second stage to obtain a SCAD estimator with the oracle property.

An \texttt{R} package for the OEM algorithm has been released. The algorithm can be speeded up by using various methods from the EM literature (McLachlan and Krishnan
2008). For example, following the idea in Varadhan and Roland (2008), one can replace the OEM iteration in (\ref{mQ}) by
$$\vb^{(k+1)}=\vb^{(k)}-2\gamma\v{r}+\gamma^2\v{v},$$where
$\v{r}=\v{M}(\vb^{(k)})-\vb^{(k)},\ \v{v}=\v{M}(\v{M}(\vb^{(k)}))-\v{M}(\vb^{(k)})-r$, and $\gamma=-\|\v{r}\|/\|\v{v}\|$. This scheme is found to lead to significant
reduction of the running time in several examples. For problems with very large $p$, one may consider a hybrid algorithm to combine the OEM and coordinate descent ideas.
It partitions $\vb$ in (\ref{lm}) into $G$ groups and in  each iteration, it minimizes the objective function $l$ in (\ref{l})  by using the OEM algorithm with respect
to one group while holding the other groups fixed. Here are some details. Group $\vb$ as $\vb=(\vb_1',\ldots,\vb_G')'$. For $k=0,1,\ldots$,
solve\begin{equation}\vb_g^{(k+1)}=\arg\min_{{\scriptsize\vb_g}}l(\vb_1^{(k+1)},\ldots,\vb_{g-1}^{(k+1)},\vb_g,\vb_{g+1}^{(k)},\ldots,\vb_G^{(k)}) \ \text{for}\
g=1,\ldots,G\label{h}\end{equation}by OEM until convergence. Note that (\ref{h}) has a much lower dimension than the iteration in (\ref{mQ}). For $G=1$, the hybrid
algorithm reduces to the OEM algorithm and for $G=p$, it becomes the coordinate descent algorithm. Theoretical properties of this hybrid algorithm will be studied and
reported elsewhere.

\vspace{4mm}
\begin{center}{\large APPENDIX: PROOFS OF THEOREM \ref{th:orlarp} AND \ref{th:orp}}\end{center}

Here are additional definitions and notation. Let $\Phi$ be the cumulative distribution function of the standard normal random variable. For $a>2$ and $\lambda$ in
(\ref{scadP}) and $d_n\geqslant\gamma_1$ in Assumption~\ref{as:Z}, define
$$s(u;\lambda)
=\left\{\begin{array}{ll}{\mathrm{sign}}(u)\big(|u|-\lambda\big)_+/d_n,&\text{when}\ |u|\leqslant(d_n+1)\lambda,
\\{\mathrm{sign}}(u)\big\{(a-1)|u|-a\lambda\big\}/\big\{(a-1)d_n-1\big\},\quad&\text{when}\ (d_n+1)\lambda<|u|\leqslant a d_n\lambda,
\\ u/d_n,&\text{when}\ |u|>a d_n\lambda,\end{array}\right.$$
and\begin{equation*}\v{s}(\v{u};\lambda)=\big[s(u_1;\lambda),\ldots,s(u_p;\lambda)\big]. \end{equation*} The OEM sequence from (\ref{scad}) satisfies the condition that
$\vb^{(k+1)}=\v{s}(\v{u}^{(k)};\lambda_n/n)$, where
\begin{eqnarray}\label{u}\v{u}^{(k)}=(\v{u}^{(k)'}_1,\v{u}^{(k)'}_2)'=\frac{\m{X}'\v{y}}{n}+\left(d_n\m{I}_p-\frac{\m{X}'\m{X}}{n}\right)\vb^{(k)}.\end{eqnarray}
For $k=1,2,\ldots$, define two sequences of events $A_k=\{\vb_2^{(k)}=\v{0}\}$ and $B_k=\{\vb_1^{(k)}=\v{u}_1^{(k-1)}/d_n\}$. Without loss of generality, assume
$\sigma^2=1$ in (\ref{lm}).

\vspace{2mm} \noindent\emph{Proof of Theorem \ref{th:orlarp}}. Since $\hat{\vb}^{\mathrm{f}}$ is a fixed point, $\hat{\vb}^{\mathrm{f}}=\v{s}(\hat{\v{u}};\lambda_n/n)$,
where $\hat{\v{u}}=(\hat{u}_1,\ldots,\hat{u}_p)'={\m{X}'\v{y}}/{n}+\left(d_n\m{I}_p-{\m{X}'\m{X}}/{n}\right)\hat{\vb}^{\mathrm{f}}$. Therefore,
\begin{eqnarray}\label{uh}\frac{\hat{\v{u}}}{d_n}=\vb+\frac{\m{X}'\ve}{nd_n}+\left(\m{I}_p-\frac{\m{X}'\m{X}}{nd_n}\right)(\hat{\vb}^{\mathrm{f}}-\vb).\end{eqnarray}
Thus,
\begin{eqnarray}&&\P(\hat{\vb}_1^{\mathrm{f}}=\hat{\v{u}}_1/d_n,\ \hat{\vb}_2^{\mathrm{f}}=\v{0})\nonumber
\\&=&\P\left(|\hat{u}_j|>ad_n\lambda_n/n\ \text{for}\ j=1,\ldots,p_1,\ |\hat{u}_j|<\lambda_n/n\ \text{for}\ j=p_1+1,\ldots,p\right)\nonumber
\\&\geqslant&1-\sum_{j=1}^{p_1}\P(|\hat{u}_j|\leqslant ad_n\lambda_n/n)-\sum_{j=p_1+1}^{p}\P(|\hat{u}_j|\geqslant \lambda_n/n).\label{p2}\end{eqnarray}By (\ref{uh})
and the fact that $\hat{\vb}_2^{\mathrm{f}}$ is concentratively consistent to $\vb$, for $j=1,\ldots,p_1$,
$\hat{u}_j/d_n=\beta_j+O_p\big((\sqrt{n}d_n)^{-1}\big)+O_p(1/c_n)=\beta_j+o_p(1)$. By $\lambda_n/n\to0$ in Assumption \ref{as:lam}, $\P(|\hat{u}_j|\leqslant
ad_n\lambda_n/n)\to0$. For the other part in (\ref{p2}),
\begin{eqnarray*}&&\sum_{j=p_1+1}^{p}\P(|\hat{u}_j|\geqslant \lambda_n/n)\\&\leqslant&\sum_{j=p_1+1}^{p}\P\left(|\v{x}_i'\ve/\sqrt{n}|\geqslant
\lambda_n/\sqrt{n}-\sqrt{n}d_n\big\|\left(\m{I}_p-(\m{X}'\m{X})/(nd_n)\right)(\hat{\vb}^{\mathrm{f}}-\vb)\big\|\right)
\\&\leqslant&2p\,\big[1-\Phi\big(\lambda_n/(2\sqrt{n})\big)\big]+p\,\P\big(c_n\|\hat{\vb}^{\mathrm{f}}-\vb\|\geqslant \lambda_nc_n/(2{n}d_n)\big)
\\&=&o\left(p\exp\left[-(\lambda_n/\sqrt{n})^2/8\right]\right)+O\left(p\exp\left(-\delta(c_n\lambda_n/(nd_n))^{\kappa}\right)\right).\end{eqnarray*}
By Assumption \ref{as:lam}, \begin{eqnarray}\P(\hat{\vb}_1^{\mathrm{f}}=\hat{\v{u}}_1/d_n,\ \hat{\vb}_2^{\mathrm{f}}=\v{0})\to1,\label{pab}\end{eqnarray} and (i) is
proved.

Now consider (ii). Note that when $\hat{\vb}_1^{\mathrm{f}}=\hat{\v{u}}_1/d_n$ and $\hat{\vb}_2^{\mathrm{f}}=\v{0}$,
$$\hat{\vb}_1^{\mathrm{f}}=\frac{\hat{\v{u}}_1}{d_n}=\vb_1+\frac{\m{X}_1'\ve}{nd_n}+\left(\m{I}_{p_1}
-\frac{\m{X}_1'\m{X}_1}{nd_n}\right)(\hat{\vb}_1^{\mathrm{f}}-\vb_1),$$ which implies that $$\m{X}_1'\m{X}_1(\hat{\vb}_1^{\mathrm{f}}-\vb_1)=\m{X}_1'\ve\sim N(\v{0},\
\sigma^2\m{X}_1'\m{X}_1).$$By (\ref{pab}) and Assumption \ref{as:Z}, the proof of (ii) is
completed.$\quad\quad\quad\quad\quad\quad\quad\quad\quad\quad\quad\quad\quad\quad\square$

To prove Theorem \ref{th:orp}, we need several lemmas.

\begin{lemma} For $k=1,2,\ldots$, if $A_k$ occurs, then
\begin{eqnarray}&&\v{u}_1^{(k)}=d_n\vb_1^{(k)}+\frac{\m{X}_1'\m{X}_1}{n}[\vb_1-\vb_1^{(k)}]+\frac{\m{X}_1'\ve}{n},\label{u1}
\\&&\v{u}_2^{(k)}=\frac{\m{X}_2'\m{X}_1}{n}[\vb_1-\vb_1^{(k)}]+\frac{\m{X}_2'\ve}{n}.\label{u2}\end{eqnarray}\label{lemma: A}\end{lemma}

\vspace{-6mm} \noindent\emph{Proof}. If $A_{k}$ occurs, then by (\ref{u}),\begin{eqnarray*}\v{u}^{(k)}&=&\frac{\m{X}'\m{X}\vb}{n}+\frac{\m{X}'\ve}{n}
+\left(\begin{array}{cc}d_n\m{I}_{p_1}-\m{X}_1'\m{X}_1/n&-\m{X}_1'\m{X}_2/n\\-\m{X}_2'\m{X}_1/n&d_n\m{I}_{p-p_1}-\m{X}_2'\m{X}_2/n\end{array}\right)
\left(\begin{array}{c}\vb_1^{(k)}\\\v{0}\end{array}\right),\end{eqnarray*}which implies the lemma.
$\quad\quad\quad\quad\quad\quad\quad\quad\quad\quad\quad\quad\quad\quad\square$

\begin{lemma} Suppose that Assumption \ref{as:Z} holds. For $k=1,2,\ldots$, if $A_1,\ldots,A_{k-1}, B_1,\ldots,B_k$ all occur, then for sufficiently large $n$,
$$\|\vb_1^{(k)}-\vb_1\|\leqslant\|\vb_1^{(0)}-\vb_1\|
+C\|\m{X}_1'\ve\|/n,$$where $C>0$ is a constant.\label{lemma: b}\end{lemma}

\noindent\emph{Proof}. If $B_1$ occurs, by (\ref{u}),
$\vb_1^{(1)}=\v{u}_1^{(0)}/d_n=\m{X}_1'\m{X}_1\vb_1/(nd_n)+\m{X}_1'\ve/(nd_n)+\big(\m{I}_{p_1}-\m{X}_1'\m{X}_1/(nd_n)\big)\vb_1^{(0)}-\m{X}_1'\m{X}_2\vb_2^{(0)}/(nd_n)$,
which implies
\begin{eqnarray*}\|\vb_1^{(1)}-\vb_1\|&=&\left\|\left(\m{I}_{p_1}-\frac{\m{X}_1'\m{X}_1}{nd_n}\right)(\vb_1^{(0)}-\vb_1)-\m{X}_1'\m{X}_2\vb_2^{(0)}/(nd_n)
+\m{X}_1'\ve/(nd_n)\right\|
\\&\leqslant&\left\|\left(\m{I}_{p_1}-\frac{\m{X}_1'\m{X}_1}{nd_n}\right)(\vb_1^{(0)}-\vb_1)-\frac{\m{X}_1'\m{X}_2}{nd_n}(\vb_2^{(0)}-\vb_2)
\right\|+\|\m{X}_1'\ve\|/(nd_n)
\\&\leqslant&\left\|\left(\m{I}_p-\frac{\m{X}'\m{X}}{nd_n}\right)(\vb^{(0)}-\vb)\right\|+\|\m{X}_1'\ve\|/(nd_n)
\\&\leqslant&\|\vb^{(0)}-\vb\|+\|\m{X}_1'\ve\|/(nd_n).\end{eqnarray*}If $A_1$, $B_1$, and $B_2$ all occur, by Lemma \ref{lemma: A},
we have $\vb_1^{(2)}=\v{u}_1^{(1)}/d_n=\vb_1^{(1)} +\m{X}_1'\m{X}_1(\vb_1-\vb_1^{(1)})/(nd_n)+\m{X}_1'\ve/(nd_n)$. Therefore,
\begin{eqnarray*}\|\vb_1^{(2)}-\vb_1\|&=&\left\|\left(\m{I}_{p_1}-\frac{\m{X}_1'\m{X}_1}{nd_n}\right)(\vb_1^{(1)}-\vb_1)+\m{X}_1'\ve/(nd_n)\right\|
\\&\leqslant&\eta_n\|\vb^{(1)}-\vb\|+\|\m{X}_1'\ve\|/(nd_n)
\\&\leqslant&\eta_n\|\vb_1^{(0)}-\vb_1\|+(1+\eta_n)\|\m{X}_1'\ve\|/(nd_n).\end{eqnarray*}Similarly, if $A_1, A_2, B_1, B_2$, and $B_3$ all occur,
we can obtain
$$\|\vb_1^{(3)}-\vb_1\|\leqslant\eta_n^2\|\vb_1^{(0)}-\vb_1\|+(1+\eta_n+\eta_n^2)\|\m{X}_1'\ve\|/(nd_n).$$
By recursion, if $A_1,\ldots,A_{k-1}, B_1,\ldots,B_k$ all occur, we have\begin{eqnarray*}\|\vb_1^{(k)}-\vb_1\|&\leqslant&\eta_n^{k-1}\|\vb_1^{(0)}-\vb_1\|
+\frac{1-\eta_n^k}{1-\eta_n}\cdot\frac{\|\m{X}_1'\ve\|}{nd_n}\\&\leqslant&\|\vb_1^{(0)}-\vb_1\| +\frac{\|\m{X}_1'\ve\|}{n(1-\eta_n)d_n}.\end{eqnarray*}This proof can be
completed by noting $(1-\eta_n)d_n$ tends to the smallest eigenvalue of $\m{\Sigma}_1$ as $n\to \infty$ from Assumption \ref{as:Z}.
$\quad\quad\quad\quad\quad\quad\quad\quad\quad\quad\quad\quad\quad\quad\square$

\begin{lemma} For $k=1,2,\ldots$, if $A_1,\ldots,A_k, B_1,\ldots,B_{k+1}$ all occur, then
$$\frac{\m{X}_1'\m{X}_1}{n}(\vb_1^{(k)}-\vb_1)=\frac{\m{X}_1'\ve}{n}+O_p(d_n\eta_n^{k}/\sqrt{n}).$$\label{lemma: Op}\end{lemma}

\vspace{-6mm} \noindent\emph{Proof}. If $A_k$ and $B_{k+1}$ both occur, by Lemma \ref{lemma: A},
$\vb_1^{(k+1)}=\vb_1^{(k)}+\m{X}_1'\m{X}_1(\vb_1-\vb_1^{(k)})/(nd_n)+\m{X}_1'\ve/(nd_n)$, which implies
\begin{eqnarray}\frac{\m{X}_1'\m{X}_1}{nd_n}(\vb_1^{(k)}-\vb_1)=\frac{\m{X}_1'\ve}{nd_n}+d_n(\vb_1^{(k)}-\vb_1^{(k+1)}).\label{e1}\end{eqnarray}
Similarly, if $A_{k-1}$ and $B_{k}$ both occur, we
have\begin{eqnarray}\frac{\m{X}_1'\m{X}_1}{nd_n}(\vb_1^{(k-1)}-\vb_1)=\frac{\m{X}_1'\ve}{nd_n}+d_n(\vb_1^{(k-1)}-\vb_1^{(k)}).\label{e2}\end{eqnarray} Combining
(\ref{e1}) and (\ref{e2}) gives
\begin{eqnarray*}\|\vb_1^{(k+1)}-\vb_1^{(k)}\|&=&\left\|\left(\m{I}_{p_1}-\frac{\m{X}_1'\m{X}_1}{nd_n}\right)(\vb_1^{(k)}-\vb_1^{(k-1)})\right\|
\\&\leqslant&\eta_n\|\vb_1^{(k)}-\vb_1^{(k-1)}\|.\end{eqnarray*}
By recursion and Lemma \ref{lemma: b}, if $A_1,\ldots,A_k, B_1,\ldots,B_{k+1}$ all occur, we have
\begin{eqnarray}\|\vb_1^{(k+1)}-\vb_1^{(k)}\|&\leqslant&\eta_n^k\|\vb_1^{(1)}-\vb_1^{(0)}\|\nonumber
\\&\leqslant&\eta_n^k(\|\vb_1^{(1)}-\vb_1\|+\|\vb_1^{(0)}-\vb_1\|)\nonumber\\&=&O_p(\eta_n^{k}/\sqrt{n})\label{Op}\end{eqnarray}
This lemma follows from (\ref{e1}) and (\ref{Op}). $\quad\quad\quad\quad\quad\quad\quad\quad\quad\quad\quad\quad\quad\quad\square$

\vspace{2mm} \noindent\emph{Proof of Theorem \ref{th:orp}}. By Lemma \ref{lemma: Op} and Assumption \ref{as:k}, it suffices to prove $\P\big((\cap_{i=1}^k
A_i)\cap(\cap_{i=1}^{k+1} B_i)\big)\rightarrow1$. In what follows, $C_1, C_2,\ldots$ all denote positive constants. For all $k=0,1,\ldots,$
$\v{u}^{(k)}=\m{X}'\ve/n+(d_n\m{I}_p-\m{X}'\m{X}/n)(\vb^{(k)}-\vb)+d_n\vb$. We have
\begin{eqnarray}|u_j^{(k)}|\geqslant d_n|\beta_j|-|\v{x}_j'\ve|/n-\|(d_n\m{I}_p-\m{X}'\m{X}/n)(\vb^{(k)}-\vb)\|\ \text{for}\ j=1,\ldots,p_1\label{ul}\end{eqnarray}and
\begin{eqnarray}|u_j^{(k)}|\leqslant|\v{x}_j'\ve|/n+\|(d_n\m{I}_p-\m{X}'\m{X}/n)(\vb^{(k)}-\vb)\|\ \text{for}\ j=p_1+1,\ldots,p.\label{us}\end{eqnarray}

First consider $A_1$ and $B_1$. By (\ref{ul}) and Assumption \ref{as:lam}, for $j=1,\ldots,p_1$,
\begin{eqnarray*}&&\P(|u_j^{(0)}|\leqslant ad_n\lambda_n/n)
\\&\leqslant&\P(|\v{x}_j'\ve|/(nd_n)\geqslant|\beta_j|/2-(a\lambda_n)/n)+\P(\|(\m{I}_p-\m{X}'\m{X}/(nd_n))(\vb^{(0)}-\vb)\|\geqslant|\beta_j|/2)
\\&\leqslant&2\big(1-\Phi(\sqrt{n}d_n|\beta_j|/2-ad_n\lambda_n/\sqrt{n})\big)+\P(c_{n}\|\vb^{(0)}-\vb\|\geqslant c_{n}|\beta_j|/2)
\\&\leqslant&C_1\exp(-C_2c_n^{\kappa}),\end{eqnarray*}which implies
\begin{eqnarray}\P(B_1)&=&\P(|u_j^{(0)}|>ad_n\lambda_n/n\ \text{for}\ j=1,\ldots,p_1)\nonumber
\\&\geqslant&1-\sum_{j=1}^{p_1}\P(|u_j^{(0)}|\leqslant ad_n\lambda_n/n)\nonumber
\\&\geqslant&1-C_3\exp(-C_2c_n^{\kappa}).\label{B1u}\end{eqnarray}By (\ref{us}), for $j=p_1+1,\ldots,p$,\begin{eqnarray*}&&\P(|u_j^{(0)}|>\lambda_n/n)
\\&\leqslant&\P(|\v{x}_j'\ve|/n>\lambda_n/(2n))+\P(\|(\m{I}_p-\m{X}'\m{X}/(nd_n))(\vb^{(0)}-\vb)\|>\lambda_n/(2nd_n))
\\&\leqslant&2\big(1-\Phi(\lambda_n/(2\sqrt{n}\,))\big)+\P(c_{n}\|\vb^{(0)}-\vb\|>c_n\lambda_n/(2{n}d_n))
\\&\leqslant&C_4\exp(-C_5(c_n\lambda_n/({n}d_n))^\kappa),\end{eqnarray*}which implies
\begin{eqnarray}\P(A_1)&=&\P(|u_j^{(0)}|\leqslant\lambda_n/n\ \text{for}\ j=p_1+1,\ldots,p)\nonumber
\\&\geqslant&1-\sum_{j=p_1+1}^p\P(|u_j^{(0)}|>\lambda_n/n)\nonumber\\&\geqslant&1-pC_4\exp(-C_5(c_n\lambda_n/({n}d_n))^\kappa).\label{A1u}\end{eqnarray}

Next consider $A_k$ and $B_k$ for $k>1$. By (\ref{ul}) and Lemma \ref{lemma: b},
\begin{eqnarray*}&&\P\left(\cup_{j=1}^{p_1}\{|u_j^{(k-1)}|\leqslant ad_n\lambda_n/n\}\right)
\\&\leqslant&\P\left(\cup_{j=1}^{p_1}\{|u_j^{(k-1)}|\leqslant ad_n\lambda_n/n\}\cap\{\cap_{i=1}^{k-2}A_i\}\cap\{\cap_{i=1}^{k-1} B_i\}\right)
\\&&+\,\sum_{i=1}^{k-2}[1-\P(A_i)]+\sum_{i=1}^{k-1}[1-\P(B_i)]
\\&\leqslant&\sum_{j=1}^{p_1}\P\Big(\{|\beta_j|-|\v{x}_j'\ve|/(nd_n)-\|(\m{I}_p-\m{X}'\m{X}/(nd_n))(\vb^{(k-1)}-\vb)\|
\leqslant a\lambda_n/n\}
\\&&\quad\quad\quad\cap\{\cap_{i=1}^{k-2}A_i\}\cap\{\cap_{i=1}^{k-1} B_i\}\Big)+\sum_{i=1}^{k-2}[1-\P(A_i)]+\sum_{i=1}^{k-1}[1-\P(B_i)]
\\&\leqslant&\sum_{j=1}^{p_1}\P(|\v{x}_j'\ve|/(nd_n)\geqslant|\beta_j|/2-a\lambda_n/n)+p_1\P(\|\vb_1^{(0)}-\vb_1\|+C\|\m{X}_1'\ve\|/n\geqslant|\beta_j|/2)
\\&&+\,\sum_{i=1}^{k-2}[1-\P(A_i)]+\sum_{i=1}^{k-1}[1-\P(B_i)]
\\&\leqslant&2p_1\big(1-\Phi(\sqrt{n}d_n|\beta_j|/2-ad_n\lambda_n/\sqrt{n})\big)+p_1\P(\|\vb_1^{(0)}-\vb_1\|\geqslant|\beta_j|/4)
\\&&+\,p_1\P(C\|\m{X}_1'\ve\|/n\geqslant|\beta_j|/4)+\sum_{i=1}^{k-2}[1-\P(A_i)]+\sum_{i=1}^{k-1}[1-\P(B_i)]
\\&\leqslant& p_1C_6\exp(-C_7c_n^{\kappa})+ \sum_{i=1}^{k-2}[1-\P(A_i)]+\sum_{i=1}^{k-1}[1-\P(B_i)], \end{eqnarray*}
which implies \begin{eqnarray}\P(B_k)\geqslant& 1-p_1C_6\exp(-C_7c_n^{\kappa})-\sum_{i=1}^{k-2}[1-\P(A_i)]-\sum_{i=1}^{k-1}[1-\P(B_i)].\label{Bit}\end{eqnarray}
Similarly, we can obtain
\begin{eqnarray}\P(A_k)\geqslant1-pC_{8}\exp(-C_{9}(c_n\lambda_n/({n}d_n))^\kappa)-\sum_{i=1}^{k-2}[1-\P(A_i)]-\sum_{i=1}^{k-1}[1-\P(B_i)].\label{Ait}\end{eqnarray}
By recursion from (\ref{B1u}), (\ref{A1u}), (\ref{Bit}), and (\ref{Ait}), we have
\begin{eqnarray}\P(B_k)\geqslant1-k^2C_{10}\exp(-C_{7}c_n^{\kappa})-pk^2C_{11}\exp(-C_{9}(c_n\lambda_n/({n}d_n))^\kappa),\label{Bo}\end{eqnarray}and
\begin{eqnarray}\P(A_k)\geqslant1-k^2C_{12}\exp(-C_{7}c_n^{\kappa})-pk^2C_{13}\exp(-C_{9}(c_n\lambda_n/({n}d_n))^\kappa).\label{Ao}\end{eqnarray}
By (\ref{Bo}) and (\ref{Ao}), $$\P\big((\cap_{i=1}^k A_i)\cap(\cap_{i=1}^{k+1}
B_i)\big)\geqslant1-k^3C_{14}\exp(-C_{7}c_n^{\kappa})-pk^3C_{15}\exp(-C_{9}(c_n\lambda_n/({n}d_n))^\kappa).$$By Assumption \ref{as:k}, we complete this proof.
$\quad\quad\quad\quad\quad\quad\quad\quad\quad\quad\quad\quad\quad\quad\square$

\vspace{10mm} \centerline{\bf ACKNOWLEDGEMENTS} Xiong is partially supported by grant 11271355 of the National Natural Science Foundation of China. Dai is partially support by
Grace Wahba through NIH grant R01 EY009946, ONR grant N00014-09-1-0655 and NSF grant DMS-0906818. Qian is partially supported by NSF grant DMS 1055214. The authors thank Jin Tian for useful discussions, and thank
Xiao-Li Meng, Grace Wahba, the editor, associate editor, and two referees for their comments and suggestions which lead to improvements in the article.

\vspace{1cm}

\begin{center}{\large REFERENCES}\end{center}

\begin{description}

\item {}
Ben-Israel, A. and Greville, T. N. E. (2003), \textit{Generalized Inverses, Theory and Applications}, 2nd ed., New York: Springer.


\item{}
Bondell, H. D. and Reich, B. J. (2008), ``Simultaneous Regression Shrinkage, Variable Selection and Clustering of Predictors With Oscar," \textit{Biometrics} 64,
115--123.

\item{}
Breiman, L. (1995), ``Better Subset Regression Using the Nonnegative Garrote," \textit{Technometrics}, 37, 373--384.

\item{}
Breheny, P. and Huang, J. (2011), ``Coordinate Descent Algorithms for Nonconvex Penalized Regression, With Applications to Biological Feature Selection," \textit{The
Annals of Applied Statistics}, 5, 232--253.

\item{}
B\"{u}hlmann, P. and van de Geer, S. (2011). \textit{Statistics for High-Dimensional Data: Methods, Theory and Applications}, Berlin: Springer.






\item{}
Dempster, A. P., Laird, N. M. and Rubin, D. B. (1977),
``Maximum Likelihood from Incomplete Data via the EM Algorithm," \textit{Journal of the Royal Statistical Society, Ser. B}, 39, 1--38.


\item{}
Efron, B., Hastie, T., Johnstone, I. and Tibshirani, R. (2004), ``Least Angle Regression,"
\textit{The Annals of Statistics}, 32, 407--451.


\item{}
Fan, J. and Li, R. (2001), ``Variable Selection via Nonconcave Penalized Likelihood and Its Oracle Properties,"
\textit{Journal of the American Statistical Association}, 96, 1348--1360.

\item{}
Fan, J. and Lv, J. (2008)
``Sure Independence Screening for Ultrahigh Dimensional Feature Space (with discussion),"
\textit{Journal of the Royal Statistical Society, Ser. B}, \textbf{70}, 849--911.

\item{}
Fan, J. and Lv, J. (2011). ``Properties of Non-concave Penalized Likelihood with NP-dimensionality," \textit{Information Theory, IEEE Transactions}, 57, 5467--5484.

\item{}
Fan, J. and Peng, H. (2004), ``Non-concave Penalized Likelihood With Diverging Number of Parameters," \textit{The Annals of Statistics}, 32, 928--961.

\item{}
Frank, L. E. and Friedman, J. (1993), ``A Statistical View of Some Chemometrics Regression Tools," \textit{Technometrics}, \textbf{35}, 109--135.

\item{}
Friedman, J., Hastie, T., Hofling, H. and Tibshirani, R. (2007), ``Pathwise Coordinate Optimization," \textit{The Annals of Applied Statistics}, 1, 302--332.

\item{}
Friedman, J., Hastie, T. and Tibshirani, R. (2009), ``Regularization Paths for Generalized Linear Models via Coordinate Descent,"
\textit{Journal of Statistical Software}, 33, 1--22.

\item{}
Friedman, J., Hastie, T. and Tibshirani, R. (2011), ``Glmnet," R package.

\item{}
Fu, W. J. (1998), ``Penalized Regressions: The Bridge versus the LASSO," \textit{Journal of Computational and Graphical Statistics}, 7, 397--416.

\item {}
Golub, G. H. and Van Loan, C. F. (1996), \textit{Matrix computations}, 3rd ed., Baltimore: The Johns Hopkins University Press.

\item{}
Grandvalet, Y. (1998), ``Least Absolute Shrinkage is Equivalent to Quadratic Penalization," \textit{In: Niklasson, L., Bod\'{e}n, M., Ziemske, T. (eds.), ICANN'98. Vol.
1 of Perspectives in Neural Computing}, Springer, 201--206.

\item{}
Green, P. J. (1990), ``On Use of the EM Algorithm for Penalized Likelihood Estimation," \textit{Journal of the Royal Statistical Society, Ser. B}, 52, 443--452.

\item{}
Hastie, T. and Efron, B. (2011), ``Lars," R package.

\item{}
Healy, M. J. R. and Westmacott, M. H. (1956), ``Missing Values in Experiments Analysed on Automatic Computers,"
\textit{Journal of the Royal Statistical Society, Ser. C}, 5, 203--206.

\item{}
Hoerl, A. E. and Kennard, R. W. (1970), ``Ridge Regression: Biased Estimation for Nonorthogonal Problems," \textit{Technometrics}, 12, 55--67.



\item{}
Hunter, D. R. and Li, R. (2005), ``Variable Selection Using MM Algorithms," \textit{The Annals of Statistics}, 33, 1617--1642.

\item{}
Huo, X. and Chen, J. (2010), ``Complexity of Penalized Likelihood Estimation," \textit{Journal of Statistical Computation and Simulation}, 80, 747--759.

\item{}
Huo, X. and Ni, X. L. (2007), ``When Do Stepwise Algorithms Meet Subset Selection Criteria?" \textit{The Annals of Statistics}, 35, 870--887.





\item{}
Mazumder, R., Friedman, J. and Hastie, T. (2011), ``SparseNet: Coordinate Descent with Non-Convex Penalties,"
\textit{Journal of the American Statistical Association}, 106, 1125--1138.

\item{}
Meinshausen, N. and Yu, B. (2009). ``Lasso-Type Recovery of Sparse Representations for High-Dimensional Data," \textit{The Annals of Statistics}, 37, 246--270.

\item{}
McLachlan, G. and Krishnan, T. (2008), \textit{The EM Algotithm and Extensions}, 2nd ed., New York: Wiley.


\item{}
Meng, X. L. (1994), ``On the Rate of Convergence of the ECM Algorithm," \textit{The Annals of Statistics}, 22, 326--339.




\item{}
NAE (2008), ``Grand Challenges for Engineering," \textit{Technical report}, The National Academy of Engineering.

\item{}
Nettleton, D. (1999), ``Convergence Properties of the EM Algorithm in Constrained Parameter Spaces," \textit{Canadian Journal of Statistics}, 27, 639--648.

\item{}
Osborne, M. R., Presnell, B. and Turlach, B. (2000), ``On the LASSO and Its Dual," \textit{Journal of Computational and Graphical Statistics}, 9, 319--337.

\item{}
Owen, A. B. (1994), ``Controlling Correlations in Latin Hypercube Samples," \textit{Journal of the American Statistical Association}, 89, 1517--1522.

\item{}
Owen, A. B. (2006), ``A Robust Hybrid of Lasso and Ridge Regression," \textit{Technical Report}.

\item{}
Petry, S. and Tutz, G. (2012), ``Shrinkage and variable selection by polytopes," \textit{Journal of Statistical Planning and Inference}, 9, 48--64.



\item{}
Schifano, E. D., Strawderman, R. and Wells, M. T. (2010), ``Majorization-Minimization Algorithms for Nonsmoothly Penalized Objective Functions,"
\textit{Electronic Journal of Statistics}, 23, 1258--1299.





\item{}
The Economist (2010), ``Special Report on the Data Deluge, (February 27)," \textit{The Economist}, 394, 3--18.

\item{}
Tibshirani, R. (1996), ``Regression Shrinkage and Selection via the Lasso,"
\textit{Journal of the Royal Statistical Society, Ser. B}, 58, 267--288.

\item{}
Tseng, P. (2001), ``Convergence of a Block Coordnate Descent Method for Nondifferentialble Minimization,"
\textit{Journal of Optimization Theory and Applications}, 109, 475--494.

\item{}
Tseng, P. and Yun, S. (2009), ``A Coordinate Gradient Descent Method for Nonsmooth Separable Minimization," \textit{Mathematical Programming B}, 117, 387--423.

\item{}
Tutz, G. and Ulbricht, J. (2009), ``Penalized Regression With Correlation-Based Penalty," \textit{Statistics and Computing}, 19, 239--253 .

\item{}
Varadhan, R. and Roland, C. (2008), ``Simple and Globally Convergent Methods for Accelerating the Convergence of Any EM Algorithm,"
\textit{Scandinavian Journal of Statistics}, 35, 335--353.


\item{}
Wang, H., Li, R., and Tsai, C-L. (2007), ``Tuning Parameter Selectors for the Smoothly Clipped Absolute Deviation Method," \textit{Biometrika}, 94, 553--568.


\item{}
Watkins, D. S. (2002), \textit{Fundamentals of Matrix Computations}, 2nd ed., New York: Wiley.

\item{}
Wilkinson, J. H. (1965), \textit{The Algebraic Eigenvalue Problem}, New York: Oxford University Press.

\item{}
Witten, D. M. and Tibshirani, R. (2011), ``Scout," R package.

\item{}
Wu, C. F. J. (1983), ``On the Convergence Properties of the EM Algorithm," \textit{The Annals of Statistics}, 11, 95--103.

\item{}
Wu, T. and Lange, K. (2008), ``Coordinate Descent Algorithm for Lasso Penalized Regression," \textit{The Annals of Applied Statistics}, 2, 224--244.


\item{}
Xu, H. (2009), ``Algorithmic Construction of Efficient Fractional Factorial Designs With Large Run Sizes," \textit{Technometrics}, 51, 262--277.


\item{}
Zangwill, W. I. (1969), \textit{Nonlinear Programming: A Unified Approach}, Englewood Cliffs, New Jersey: Prentice Hall.

\item{}
Zhang. C-H. (2010), ``Nearly Unbiased Variable Selection under Minimax Concave Penalty," \textit{The Annals of Statistics}, 38, 894--942.


\item{}
Zou, H. and Hastie, T. (2005), ``Regularization and Variable Selection via the Elastic Net,"
\textit{Journal of the Royal Statistical Society, Ser. B}, 67, 301--320.

\item{}
Zou, H. and Li, R. (2008), ``One-step Sparse Estimates in Nonconcave Penalized Likelihood Models," \textit{The Annals of Statistics}, 36, 1509--1533.

\end{description}
\end{document}